\documentclass[aps,pra,twocolumn,superscriptaddress,floatfix]{revtex4-1}
\usepackage{graphicx}  % needed for figures
\usepackage{dcolumn}   % needed for some tables
\usepackage{bm}        % for math
\usepackage{amssymb}   % for math
\usepackage{amsmath}
\usepackage{hyperref}
\usepackage{xcolor}
\usepackage{float}
\usepackage[amsmath,amsthm,thmmarks]{ntheorem}
\hypersetup{
	colorlinks,
	linkcolor={red!50!black},
	citecolor={blue!50!black},
	urlcolor={blue!80!black}
}
\newtheorem{theorem}{Theorem}

\newcommand{\ket}[1]{\ensuremath{|#1\rangle}}

\DeclareMathOperator*{\argmin}{arg\,min}
\DeclareMathOperator{\ent}{S}  %Von Neumann Entropy
\DeclareMathOperator{\cent}{S_{A|B}}  % Conditional Von Neumann Entropy
\DeclareMathOperator{\Tr}{Tr}
\DeclareMathOperator{\sent}{H} %Shannon entropy
\DeclareMathOperator{\osent}{H_A} %One sided Shannon entropy
\DeclareMathOperator{\hash}{D_{\rightarrow}} % Hashing inequality 

\DeclareMathOperator{\func}{f}
\DeclareMathOperator{\qcap}{Q}
\DeclareMathOperator{\Free}{{\cal F_{\mathrm v} (H)}}
\DeclareMathOperator{\All}{{\cal S(H)}}
\newcommand{\orig}{\gamma}
\newcommand{\test}{\delta}
\newcommand{\A}{{\mathchoice{}{}{\scriptscriptstyle}{}A}}
\newcommand{\B}{{\mathchoice{}{}{\scriptscriptstyle}{}B}}

\begin{document}

%\preprint{}

\title{Witnessing Negative Conditional Entropy}


\author{Mahathi Vempati}
\email[]{mahathi.vempati@research.iiit.ac.in}
%\thanks{}
\affiliation{Center for Computational Natural Sciences and Bioinformatics, International Institute of Information Technology-Hyderabad, Gachibowli, Telangana-500032, India.}
\author{Nirman Ganguly}
\email[]{nirmanganguly@hyderabad.bits-pilani.ac.in}
%\thanks{}
\affiliation{Department of Mathematics, Birla Institute of Technology and Science Pilani, Hyderabad Campus, Telangana-500078, India.}
\author{Indranil Chakrabarty}
\email[]{indranil.chakrabarty@iiit.ac.in}
%\thanks{}
\affiliation{Center for Security, Theory and Algorithmic Research, International Institute of Information Technology-Hyderabad, Gachibowli, Telangana-500032, India.}
\affiliation{Quantum Information and Computation Group,
	Harish-Chandra Research Institute, HBNI, Allahabad 211019, India.}
\author{Arun K Pati}
\email[]{akpati@hri.res.in}
%\thanks{}
\affiliation{Quantum Information and Computation Group,
	Harish-Chandra Research Institute, HBNI, Allahabad 211019, India.}


\date{\today}

\begin{abstract}
    Quantum states that possess negative conditional von Neumann entropy provide quantum advantage in several information-theoretic protocols including superdense coding, state merging, distributed private randomness distillation and one-way entanglement distillation. While entanglement is an important resource, only a subset of entangled states have negative conditional von Neumann entropy. In this work, we characterize the class of density matrices having non-negative conditional von Neumann entropy as convex and compact. This allows us to prove the existence of a Hermitian operator (a witness) for the detection of states having negative conditional entropy for bipartite systems in arbitrary dimensions. We show two constructions of such witnesses. For one of the constructions, the expectation value of the witness in a state is an upper bound to the conditional entropy of the state. We pose the problem of obtaining a tight upper bound to the set of conditional entropies of states in which an operator gives the same expectation value. We solve this convex optimization problem numerically for a two qubit case and find that this enhances the usefulness of our witnesses. We also find that for a particular witness, the estimated tight upper bound matches the value of conditional entropy for Werner states. We explicate the utility of our work in the detection of useful states in several protocols. 
\end{abstract}

%\keywords{}

\maketitle

\section{Introduction}
\label{intro}
Entanglement is an indispensable resource for several information processing tasks \cite{enthorodecki}. However, not all entangled states qualify to be a resource in some information processing tasks. For example, under the standard teleportation scheme, only entangled states whose fully entangled fraction is above a threshold value are considered to be useful \cite{telwit}. Therefore, even within the set of entangled states one needs to identify distinctive features from an operational perspective.
\par One such distinctive feature is given by conditional von Neumann entropy, also known as the quantum conditional entropy \cite{negentropy}. The quantum conditional entropy (henceforth referred to as conditional entropy when there is no ambiguity) for a quantum state $ \rho_{AB} $ is defined as
\begin{equation}
\cent (\rho_{AB}) = \ent(\rho_{AB}) - \ent(\rho_B),
\end{equation} where $ \ent(\rho)=-\Tr(\rho \log \rho) $ is the von Neumann entropy of the state. Unlike its classical counterpart, conditional entropy in the quantum realm can be negative \cite{negentropy} and thus can be exploited for quantum information tasks \cite{negentropy, non-additivity}. Its operational interpretation is given in the context of quantum state merging \cite{statemerging1,statemerging2}, where negative conditional entropy is an indication of resources for future communication \cite{statemerging1,statemerging2}. Negative conditional entropy provides quantum advantage in superdense coding \cite{sdc1,sdc2,sdc3} and  is a characteristic of states for which one-way entanglement distillation is possible \cite{one_way_entangle}. It also helps in the reduction in the uncertainty in predicting the outcomes of two incompatible measurements \cite{ur1} and in the maximization of  the rates of distributed private randomness distillation \cite{private_rand}.

The significance of negative conditional entropy has a wide-ranging impact beyond quantum information theory.  In \cite{blackholes_entropy, blackholes_second_law}, a modification of the Bekenstein-Hawking area law for a Schwarzchild black hole that resolves previously present paradoxes is shown, where simple entropy in the original law is replaced with the negative of conditional entropy of the purifying system of the black hole with respect to the positive-energy particles in the black hole. According to the modified law, when the change in the above-mentioned conditional entropy is negative, there is an increase in the surface area of the black hole. In the context of thermodynamics, it is shown that the work cost of erasing a system is proportional to the conditional entropy of the state of the system with respect to that of the observer \cite{thermo}. In particular, when this conditional entropy is negative, the observer gains work when erasing the system, thereby cooling the environment rather than heating it up.

While all states possessing negative conditional entropy are entangled, the converse is not true. This fact, along with the vast range of applications like those mentioned above warrants that quantum states are characterized with respect to their conditional entropy. To this end, in \cite{patro}, we have characterized states whose conditional entropy remains non-negative even after the application of global unitary operations. Conditional entropy is also not a direct observable, and this calls for practical methods to detect states that possess this resource. 

\par This paper is organized as follows: In Section \ref{sec:exist}, we characterize the class of states with non-negative conditional entropy as convex and compact. This allows for the existence of witnesses for all states outside this class, and we construct such a witness. We also discuss examples of this witness for Werner and Isotropic states. In Section \ref{sec:upper}, we prove that the expectation value of any of the witnesses in a state is an upper bound to the state's conditional entropy. We then pose the problem of obtaining a tight upper bound to the set of conditional entropies of states in which an operator gives the same expectation value. We solve this convex optimization problem numerically, and discuss several interesting results. In Section \ref{sec:geo}, we provide an alternative witness formulation exploiting geometric methods. In Section \ref{sec:app}, we discuss some applications of our work. In Section \ref{sec:conclusion}, we conclude.

\begin{figure}
	\includegraphics[width=\linewidth]{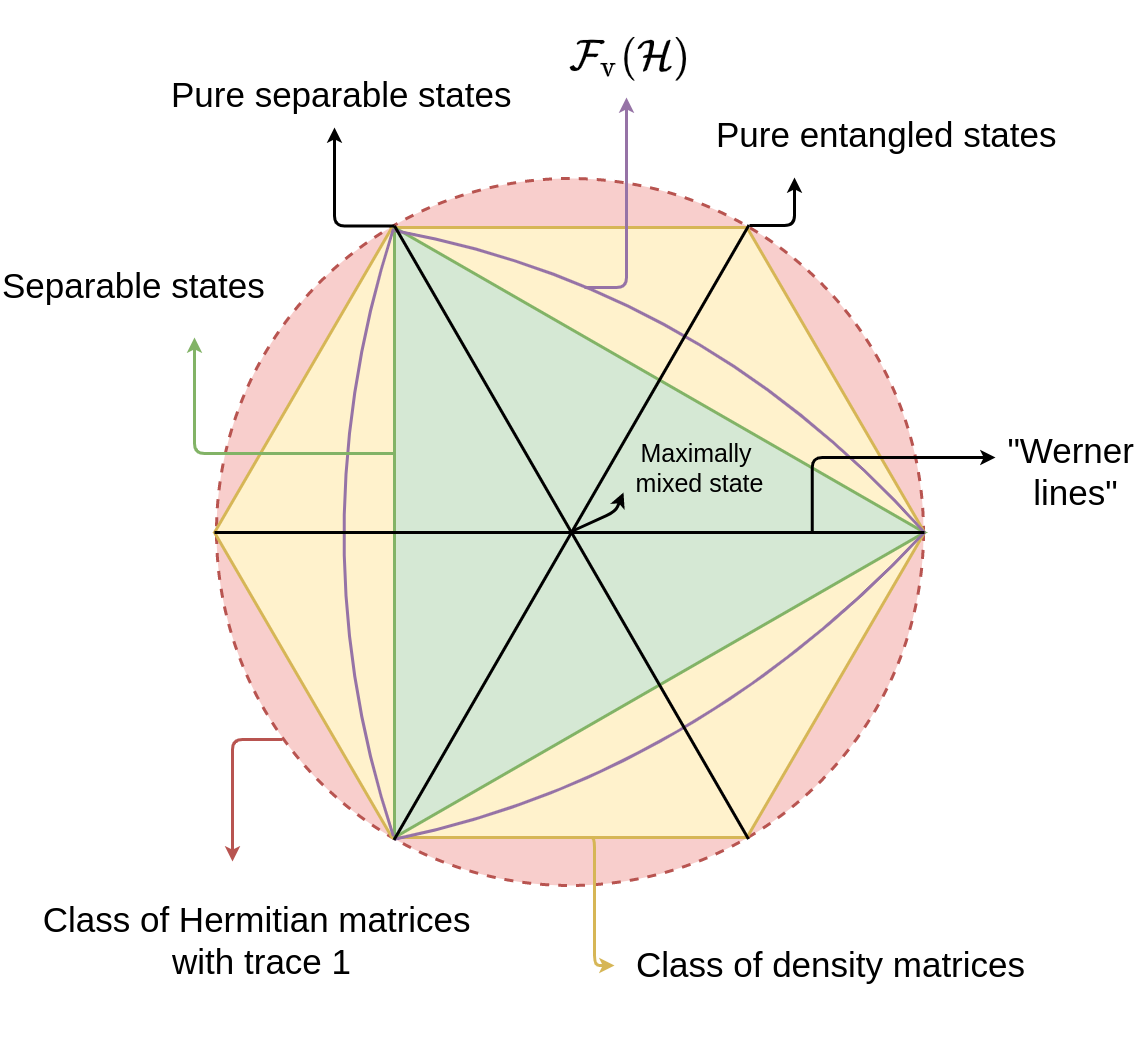}
	\caption{Schematic diagram of $\cal{F}_{\mathrm v} (\cal H)$: The outer class is the set of all Hermitian matrices with Trace 1. It extends beyond the figure as indicated by the dotted lines. Within that is the class of density matrices. The pure states (entangled or separable) lie farthest from the maximally mixed state. The density matrix set is the set of all convex combinations of pure state projectors -- but since there are infinite such projectors, it does not form a polytope, unlike what it may seem like from the figure. Within the density matrix space is the convex class of separable states. $\cal{F}_{\mathrm v} (\cal H)$ is also a convex class and it contains the class of separable states.}
	\label{fig:cvenn}
\end{figure}
\section{Existence and construction of a witness for negative conditional entropy states}
\label{sec:exist}
\subsection{Characterization of non-negative conditional entropy states}
\label{sec:char}
In this work, the set of all quantum states (density matrices) over a Hilbert space ${\cal H}$ is denoted by ${\cal S(H)}$, and ${\cal H}$ has the dimensionality $d \otimes d$. We refer to the class of states possessing non-negative conditional entropy as $\Free$. A schematic representation of $\cal{F}_{\mathrm v} (\cal H)$ is shown in FIG. \ref{fig:cvenn}. We now prove that $\cal{F}_{\mathrm v} (\cal H)$ is convex and compact.
\begin{theorem}
\label{convex_compact_theorem}
	$\cal{F}_{\mathrm v} (\cal H)$ is convex and compact.
\end{theorem} 
\begin{proof} 
\par \textit{$\cal{F}_{\mathrm v} (\cal H)$ is convex---} 
 Let  $\sigma_1, \sigma_2 \in \cal{F}_{\mathrm v} (\cal H)$, i.e., $\cent(\sigma_1), \cent(\sigma_2) \geq 0 $. Let $ \sigma = \lambda \sigma_1 + (1-\lambda)\sigma_2 $, where $ \lambda \in [0,1] $. From the concavity of $\cent$ \cite{nielsen}, we have $\cent(\sigma) \geq \lambda \cent(\sigma_1) + (1 - \lambda)\cent(\sigma_2)$. Therefore, $\cent(\sigma) \geq 0$ and $\sigma \in \cal{F}_{\mathrm v} (\cal H)$, implying $\cal{F}_{\mathrm v} (\cal H)$ is convex. 

% \par \textit{$cal{F}_{mathrm v} (cal H)$ is convex---} 
% The conditional von Neumann entropy is a concave function \cite{nielsen}. Consider two density matrices $ \sigma_1 , \sigma_2 \in $ $cal{F}_{mathrm v} (cal H)$, i.e., $ \cent(\sigma_1), \cent(\sigma_2) \ge 0 $. Let $ \sigma = \lambda \sigma_1 + (1-\lambda)\sigma_2 $, where $ \lambda \in [0,1] $. Then, by concavity of conditional von Neumann entropy $ \cent(\sigma) \ge 0$. Hence $ \sigma \in  $ $cal{F}_{mathrm v} (cal H)$ implying $cal{F}_{mathrm v} (cal H)$ is convex.	
\par \textit{$\cal{F}_{\mathrm v} (\cal H)$ is compact---}
For $\sigma_{AB} \in \Free$ we have $\cent(\sigma_{AB}) \geq 0$ by definition. The value $0$ is attained for several states in $\Free$, for example, all pure separable states. From the subadditivity of $\cent$ \cite{nielsen}, we have $\ent(\sigma_{AB}) \leq \ent(\sigma_A) + \ent(\sigma_B)$, implying $\cent(\sigma_{AB}) \leq \ent(\sigma_A) \leq \log_2 d$. This value is attained at $I/d^2$. Thus,  the image set of $\Free$ under $\cent$ is a closed set, i.e. $\cent \left( \Free \right) = [0, \log_2d]$. As $\cent$ is a continuous function under trace norm \cite{alicki,winter}, we conclude that $\Free$  is also closed \cite{hahn}. $\Free$  is bounded under trace norm as every density matrix has a bounded spectrum. For our finite dimensional Hilbert space, this implies that $\Free$  is compact. 
% \par \textit{$cal{F}_{mathrm v} (cal H)$ is compact---} Inverse images of closed sets are closed under continuous mappings \cite{hahn}. Conditional von Neumann entropy is a continuous function under trace norm \cite{alicki,winter}. Therefore, in order to show that $cal{F}_{mathrm v} (cal H)$ is closed, it is enough to show that the image set under conditional entropy is closed.\\
% Note that by definition $ \cent \ge 0, \; \forall \;  \sigma_{AB} \in $ $cal{F}_{mathrm v} (cal H)$. Now, it is known that $\cent(\sigma_{AB}) = \ent(\sigma_A) - I(A:B)$ where $I(A:B)$ is the quantum mutual information between $\sigma_A$ and
% $\sigma_B$, and is a non-negative quantity. Entropy of a $d$-dimensional system is upper bounded by $\log_2 d$, so we have $\ent(\sigma_A) \leq \log_2 d$. Therefore, $ \cent(\sigma_{AB}) \in [0,\log_2 d] $.	Hence, $cal{F}_{mathrm v} (cal H)$, i.e., $ S^{-1}_{A|B}(\sigma_{AB}) $ is closed. $cal{F}_{mathrm v} (cal H)$ is bounded under trace norm as every density matrix has a bounded spectrum. Since we are considering finite dimensional Hilbert spaces, $cal{F}_{mathrm v} (cal H)$ is compact.
\end{proof}	
By the Hahn-Banach theorem \cite{hahn}, this convexity and compactness of $\cal{F}_{\mathrm v} (\cal H)$  implies that for every state with negative conditional entropy, there exists a Hermitian operator that separates it from $\cal{F}_{\mathrm v} (\cal H)$  (known as a witness). We provide an explicit expression for a family of such witnesses below.

\subsection{Witness to detect negative conditional entropy}
\label{sec:witness}
To prove that a Hermitian operator $W$ is a witness for states possessing negative conditional entropy, it is sufficient to prove (i) $\exists \, \rho \notin \Free  \mid \Tr(W \rho) < 0$, and (ii) $\forall \, \sigma \in \Free   , \, \, \Tr(W \sigma) \geq 0$.
\begin{theorem}
	\label{witness}
	A witness for any state $ \rho_{AB} \notin {\cal F_{\mathrm v} (H)} $ is\\
	$ W_\rho  = - \log(\rho_{AB}) + I \otimes \log(\rho_B).$
\end{theorem}
\begin{proof}
    $ W_\rho $ is Hermitian due to the Hermiticity of $\log(\rho_{AB})$ and $\log(\rho_B)$. To prove (i), we have $\Tr(W_\rho \rho_{AB}) = \cent(\rho_{AB}) < 0$. To prove (ii), using the monotonicity of the relative entropy \cite{nielsen, wilde} we have $\forall \, \sigma_{AB} \in \Free, \ent(\sigma_B || \rho_B) \leq \ent(\sigma_{AB} || \rho_{AB})$. On substituting the expression for relative entropy, we have $ \Tr(\sigma_B \log \sigma_B) -  \Tr(\sigma_B \log \rho_B) \leq \Tr(\sigma_{AB} \log \sigma_{AB}) - \Tr(\sigma_{AB} \log \rho_{AB})$. On rearranging, we have $-\Tr(\sigma_{AB} \log \rho_{AB}) +  \Tr(\sigma_B \log \rho_B) \geq \cent(\sigma_{AB}) \geq 0$ which implies that $\Tr(W_\rho \sigma_{AB}) \geq 0$. 
Hence, $W_\rho$ is a witness operator. 
\end{proof}

% \begin{widetext}
% 	\begin{eqnarray}
% 	&& \Tr(\sigma_B log \sigma_B) - \Tr(\sigma_B log \rho_B) 
% 	 \leq \Tr(\sigma log \sigma ) - \Tr(\sigma log \rho_{AB}) \nonumber \\
% 	&& - \Tr(\sigma log \sigma ) + \Tr(\sigma_B log \sigma_B) 
% 	\leq - \Tr(\sigma log \rho_{AB}) + \Tr(\sigma_B log \rho_B) \nonumber \\
% 	&& - \Tr(\sigma log \rho_{AB}) + \Tr(\sigma_B log \rho_B) 
% 	 \geq \cent(\sigma) \nonumber \\
% 	&&- \Tr(\sigma log \rho_{AB} + \Tr(\sigma_B log \rho_B) 
% 	 \geq 0 \nonumber \\
%     &&Tr \left(\sigma \left(- log \left(\rho_{AB}\right) + I \otimes log \left(\rho_B \right)\right)\right) 
%      \geq 0 \nonumber \\
% 	&& \therefore \; \Tr(W  \sigma) 
% 	\geq 0.
% 	\end{eqnarray}
% \end{widetext}	

The witness given in Theorem \ref{witness} is closely related to the conditional amplitude operator \cite{negentropy, condamp, cond_geo}. Note that this witness can only be constructed for full-rank density matrices \footnote{This is because $\log 0$ is undefined}. In Sec \ref{sec:geo}, we have also given an alternative prescription of a witness exploiting geometrical considerations that can be constructed for any state. We also note that it is possible to construct witnesses employing the uncertainty principle in the presence of quantum memory \cite{ur1}. We now provide examples of witness construction using Theorem \ref{witness}. 

\subsection{Examples}
\subsubsection{Werner states}
Let us consider the Werner state \cite{werner} $\orig=0.99 | \phi^+ \rangle \langle \phi^+| + 0.01\frac{I}{4}$. According to the prescription in Theorem \ref{witness} we construct a Hermitian operator $W_\orig = - \log(\orig_{\A\B}) + I \otimes \log(\orig_\B)$ which is given by
\begin{equation}
\label{wernerwit}
W_\orig = \begin{bmatrix}
a & 0 & 0 & c \\
0 & b & 0 & 0 \\
0 & 0 & b & 0 \\
c & 0 & 0 & a \\
\end{bmatrix},
\end{equation}
where $a \approx3.3274$, $b \approx 7.6439 $ and $c \approx-4.3165 $. 
We thus obtain
\begin{equation}
\label{expectation_value_werner}
\Tr(W_\orig \orig) \approx -0.9244.
\end{equation}
We observe that the witness gives us a negative value, indicating that $\orig $ possesses negative conditional entropy. 
Consider another state that $W_\orig$ witnesses: $\test=0.9 | \phi^+ \rangle \langle \phi^+| + 0.1\frac{I}{4}$. We have 
\begin{equation}
    \label{rho_test}
    \Tr(W_\orig \test) \approx -0.3417.
\end{equation} 
Once again, this certifies the negative conditional entropy. For implementation in a laboratory, it is important that the witness is decomposed in terms of local observables \cite{guhne}. If we consider the witness (\ref{wernerwit}), then we obtain the decomposition in terms of Pauli matrices given by
\begin{equation}
W_\orig= \frac{a+b}{2} I \otimes I + \frac{a-b}{2} Z \otimes Z + \frac{c}{2} X \otimes X - \frac{c}{2} Y \otimes Y.
\label{decomposition}
\end{equation}
For example, if we have polarized photons, one may take $|H\rangle = |0\rangle$, $|V\rangle = |1\rangle$, $|D\rangle = \frac{|H\rangle + |V\rangle}{\sqrt 2}$, $|F\rangle = \frac{|H\rangle - |V\rangle}{\sqrt 2}$, $|L\rangle = \frac{|H\rangle + i|V\rangle}{\sqrt 2}$, $|R\rangle = \frac{|H\rangle - i|V\rangle}{\sqrt 2}$. Using this basis, the witness can be decomposed as follows, where the values of $a$, $b$ and $c$ are defined in Equation \ref{wernerwit}:

\begin{equation}
\begin{aligned}
W_\orig = a \; \big(|HH\rangle \langle HH| + |VV\rangle \langle VV|\big)  \\
+ \; b \; \big(|VH\rangle \langle VH| + |HV \rangle \langle HV| \big)\\ 
+ \; c \; \big(|DD\rangle \langle DD| + |FF\rangle \langle FF| \\ 
- \; |RR\rangle \langle RR| - |LL \rangle \langle LL| \big).
\end{aligned}
\end{equation}
\subsubsection{Isotropic states}
Isotropic states of dimensionality $d$ are states of the form:
\begin{equation}
\rho_{Iso}=\alpha | \Phi^+ \rangle \langle \Phi^+ | +  (1 - \alpha)\frac{I}{d^2},
\end{equation}
where $ | \Phi^+ \rangle$ is a maximally entangled state of two qudits:
\begin{equation}
\begin{split}
|\Phi^+ \rangle = \frac{1}{\sqrt d} \sum_{i=1}^d |i \rangle \otimes |i \rangle, \\
\alpha \in \left[ - \frac{1}{d^2-1}, 1\right].
\end{split}
\end{equation}
The witness for such states is given by: 
\begin{equation}
\begin{aligned}
    W_{\rho_{Iso}} &= - \log(\rho_{Iso}) + I \otimes \log(\rho_{Iso_B}) \\
&= - \log(\rho_{Iso}) + I \otimes \log \left(\frac{I}{d}\right) \\
&= -\log(\rho_{Iso}) - \log(d).I .\\
\end{aligned}
\end{equation}
Here, we take an example of a $3 \times 3 $ system. We construct the analytical witness using the following state $\rho_{i} = 0.8 | \phi^+ \rangle \langle \phi^+| + 0.2 \;\frac{I}{9}$. The witness is given by

\begin{figure}
	\includegraphics[width=\linewidth]{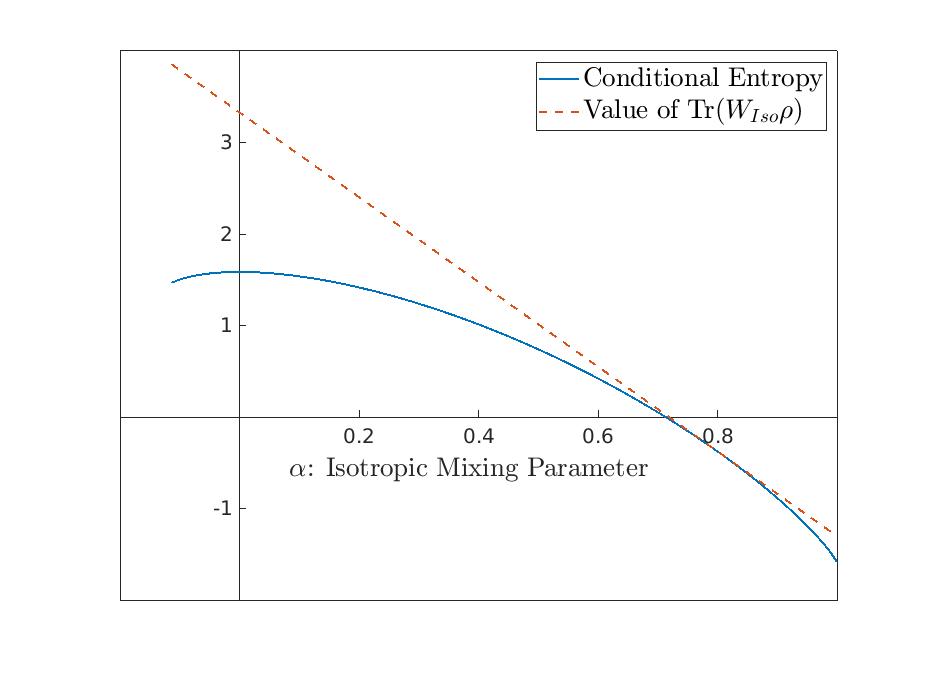}
    \caption{The $x-$axis is $\alpha$, the isotropic mixing parameter, and refers to three dimensional states of the form $\rho =  \alpha | \phi^+ \rangle \langle \phi^+ | +  (1 - \alpha)\frac{I}{d^2}$. The $y-$axis represents different quantities for the red (dashed) line, and the blue (solid) line, placed together for comparison. For the red (dashed) line, it represents the value of $\Tr(W_{\rho_i} \rho)$ and for the blue (solid) line, it represents the conditional von Neumann entropy in bits. Notice that the witness gives a positive value whenever the Conditional von Neumann Entropy is positive.}
	\label{figisotropic}
	% IN THE PLOT, MAKE IT THE ISOTROPIC MIXING PARAMETER!
	% ADD A PROPER CAPTION.
\end{figure}

\begin{equation}
\begin{aligned}
W_{\rho_i} = \;
&a \big(|00\rangle \langle 00| + |11\rangle \langle 11| + |22 \rangle \langle 22| \big)  \\
&+ b \big(|01 \rangle \langle 01| + |02\rangle \langle 02| + |10 \rangle \langle 10 | \\
&+ |12 \rangle \langle 12| + |20 \rangle \langle 20 | + |21 \rangle \langle 21 | \big)\\
&+ c \big(|00 \rangle \langle 11| + |00 \rangle \langle 22| + |11\rangle \langle 00 |  \\
& + |11\rangle \langle 22| + |22 \rangle \langle 00| + |22 \rangle \langle 11| \big), \\
\end{aligned}
\end{equation}
where $a \approx 2.1704 , b \approx 3.9069
\text{ and } c \approx -1.7365$.

We test the witness against $\rho_i$ and notice that it is negative as expected:
\begin{equation}
\Tr(W_{\rho_i}  \rho_i) \approx -0.3764.
\end{equation}
The witness does not give a negative value for all matrices with negative conditional entropy. Consider the state $\rho_t = 0.715|\phi^+ \rangle \langle \phi^+ | + 0.285 \frac{I}{9}$ which has negative conditional entropy. However, the witness above (formed from the state $\rho_i$) gives a positive value:
\begin{equation}
\Tr(W_{\rho_i}  \rho_t) \approx 0.0172.
\end{equation}
The behaviour of the above witness is shown in Fig. \ref{figisotropic} . This witness can be used to detect negative conditional entropy in a lab through the decomposition in terms of Gell Mann matrices, which are:

\begin{equation}
\begin{split}
&I = \begin{pmatrix} 1 & 0 & 0 \\ 0 & 1 & 0 \\ 0 & 0 & 1 \end{pmatrix},
\lambda_1 = \begin{pmatrix} 0 & 1 & 0 \\ 1 & 0 & 0 \\ 0 & 0 & 0 \end{pmatrix},
\lambda_2 = \begin{pmatrix} 0 & -i & 0 \\ i & 0 & 0 \\ 0 & 0 & 0 \end{pmatrix} \\
&\lambda_3 = \begin{pmatrix} 1 & 0 & 0 \\ 0 & -1 & 0 \\ 0 & 0 & 0 \end{pmatrix},
\lambda_4 = \begin{pmatrix} 0 & 0 & 1 \\ 0 & 0 & 0 \\ 1 & 0 & 0 \end{pmatrix},
\lambda_5 = \begin{pmatrix} 0 & 0 & -i \\ 0 & 0 & 0 \\ i & 0 & 0 \end{pmatrix} \\
&\lambda_6 = \begin{pmatrix} 0 & 0 & 0 \\ 0 & 0 & 1 \\ 0 & 1 & 0 \end{pmatrix},
\lambda_7 = \begin{pmatrix} 0 & 0 & 0 \\ 0 & 0 & -i \\ 0 & i & 0 \end{pmatrix},
\lambda_8 = \frac{1}{\sqrt{3}} \begin{pmatrix} 1 & 0 & 0 \\ 0 & 1 & 0 \\ 0 & 0 & -2 \end{pmatrix}. \\
\end{split}
\end{equation}
The above witness $W_{\rho_i}$ can easily be decomposed into the Gell-Mann basis, and is infact diagonal in the basis:
\begin{equation}
\begin{aligned}
W_{\rho_i} = -a\lambda_1 \otimes \lambda_1 + a\lambda_2 \otimes \lambda_2 - a\lambda_3 \otimes \lambda_3 \\ -a\lambda_4 \otimes \lambda_4 + a\lambda_5 \otimes \lambda_5 - a\lambda_6 \otimes \lambda_6  \\
+ a\lambda_7 \otimes \lambda_7 - a\lambda_8 \otimes \lambda_8 + bI_3 \otimes I_3,\\
\end{aligned}
\end{equation}
where $a \approx 0.86825 \text{ and } b \approx 3.3281$. This facilitates the identification of negativity of conditional entropy in the lab.

\section{Upper bound on conditional entropy}
\label{sec:upper}
We now use the expectation value of the witness to quantify the amount of conditional entropy in a state. To do so, we define $T_M^\chi({\cal H}) := \{ \cent(\sigma) \mid \sigma \in \All \land \Tr(M \sigma) = \chi \}$ where $M$ is a Hermitian operator and $\chi \in \mathbb{R}$. The ${\cal H}$ is henceforth omitted as there is no ambiguity. In words, $T_M^\chi$ is the set of conditional entropies of all states in which the expectation value of $M$ is $\chi$. Obtaining a tight upper bound for $T_M^\chi$ is equivalent to asking the question: what is the maximum value of conditional entropy that an unknown state can possess, given that the expectation value of $M$ in the state is $\chi$? For several applications in quantum information, as described at the end of this paper, a state is useful only if its conditional entropy is less than a certain value. Thus, an upper bound of $T_M^\chi$ could aid in the decision of whether to use a state for a given application or to discard it. In the following theorem, we prove that for a witness $W_\rho$ constructed using Theorem \ref{witness}, an upper bound of $T_{W_\rho}^\chi$ is $\chi$ itself. However, this upper bound need not be tight. We then pose the problem of obtaining a tight upper bound as a convex optimization problem, solve it numerically and discuss some interesting results. 

\begin{theorem}
	\label{upper_bound}
    $\chi$ is an upper bound of $T_{W_\rho}^\chi$.
\end{theorem}
\begin{proof}
    Let $x \in T_{W_\rho}^\chi$. Therefore, $\exists \, \sigma_{AB} \in \All$ such that $x = \cent(\sigma_{AB})$ and $\Tr(W_\rho \sigma_{AB}) = \chi$.
    Once again, using the monotonicity of relative entropy \cite{nielsen, wilde}, we have $ \ent(\sigma_B || \rho_B)  \leq \ent(\sigma_{AB} || \rho_{AB})$. On substituting for relative entropy, we have $\Tr(\sigma_B \log \sigma_B) -  \Tr(\sigma_B \log \rho_B) \leq \Tr(\sigma_{AB} \log \sigma_{AB}) - \Tr(\sigma_{AB} \log \rho_{AB})$. On rearranging, we have $ \cent(\sigma_{AB}) \leq \Tr(W_\rho \sigma_{AB})$ which implies that $x \leq \chi$.      
\end{proof}

As an example, we observe from Equation \ref{rho_test} that  $\Tr(W_\orig \test) \approx -0.3417$. According to Theorem \ref{upper_bound}, this value should serve as an upperbound to $\cent(\test)$. Indeed, we compute that $\cent(\test)$ is approximately $-0.4968$. 

In general, $\chi$ is not a tight upper bound for $T_{W_\rho}^\chi$. However when $\chi = \Tr(W_\rho \rho) = \cent(\rho)$, We have $\chi \in T_{W_\rho}^\chi$ while also being an upper bound to $T_{W_\rho}^\chi$, and is therefore a tight upper bound. From Equation \ref{expectation_value_werner} we observe that $-0.9244$ is the tight upperbound to $T_{W_\orig}^{-0.9244}$.

A tight upper bound for $T_M^\chi$ can be obtained in general by solving the following convex optimization problem for all $\sigma \in \All $: 
\begin{equation}
    \label{optimization}
\begin{aligned}
\underset{\sigma}{\operatorname{maximize}} \quad &
\cent(\sigma)\\
\textrm{subject to} \quad & \Tr(M \sigma) = \chi.\\
\end{aligned}
\end{equation} For a small number of qubits, this is easy to solve numerically \cite{cvx_1, cvx_2, quantinf, cvxquad}. We simultaneously describe the procedure as well as results of a numerical experiment to solve this optimization problem for $M = W_\orig$ defined in Equation \ref{wernerwit}. For our numerical experiment, we denote the 2-qubit Werner states with $\sigma_\alpha \equiv \alpha | \phi^+ \rangle \langle \phi^+ | + (1-\alpha) I/4$ where $\alpha \in [0,1]$.  The $x-$axis in Figure \ref{fig:werner_upper_bound} is the value of $\alpha$ for the Werner states. First, for each $\alpha$, we compute the value of $\chi(\alpha) = \Tr(W_\orig \sigma_\alpha)$, which is plotted with a red dashed line. For each $\chi(\alpha)$, we then solve the optimization problem (\ref{optimization}) to obtain the tight upper bound of $T_{W_\orig}^{\chi(\alpha)}$, which is plotted with a blue solid line. We then compute the value of $\cent(\sigma_\alpha)$, which is also the blue solid line. We observe two interesting results: Firstly, the value of the tight upper bound coincides with the value of conditional entropy for the Werner states. That is, for this particular witness $W_\orig$, among all 2-qubit states in which the expectation value of $W_\orig$ is $\chi(\alpha)$, the corresponding Werner state $\sigma_\alpha = \alpha | \phi^+ \rangle \langle \phi^+ | + (1-\alpha) I/4$ possesses the maximum conditional entropy. The second interesting aspect is that the convex optimization enhances the usefulness of the witness. We observe from the graph in Figure \ref{fig:werner_upper_bound} that there exist states where a positive $\chi(\alpha)$ is obtained, but the tight upper bound of $T_{W_\orig}^{\chi(\alpha)}$ is negative, certifying that the state indeed possesses negative conditional entropy.  

\begin{figure}
    \includegraphics[width=\linewidth]{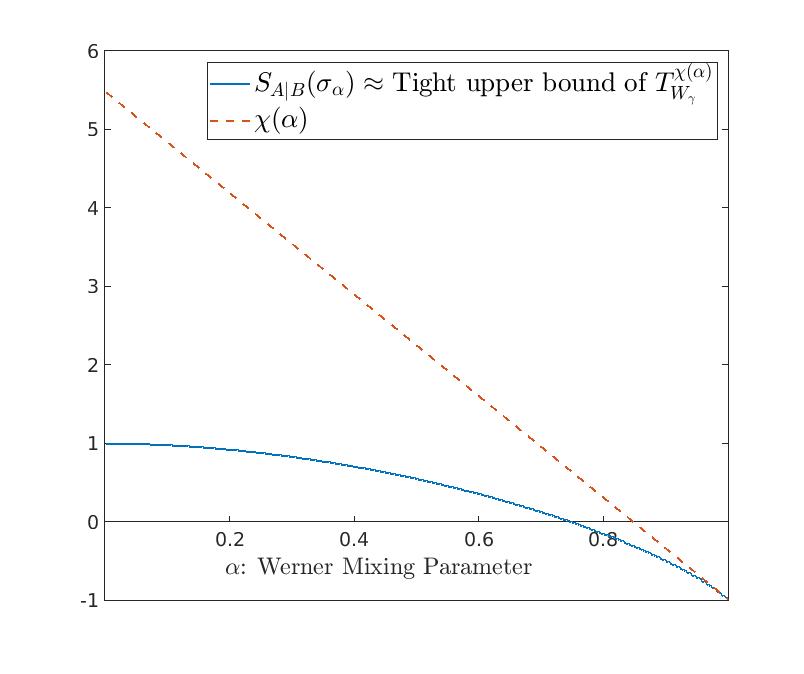}
    \caption{Conditional entropy of Werner states and the obtained tight upper bound: The $x-$axis refers to the parameter $\alpha$ for states of the form $\sigma_\alpha = \alpha | \phi^+ \rangle \langle \phi^+ | + (1-\alpha) I/4$. The $y-$axis represents different quantities for the red (dashed) line and the blue (solid) line, placed together for comparison. For the red (dashed) line, it represents the value of $\chi(\alpha) = \Tr(W_\orig \sigma_\alpha)$. For the blue (solid) line, it represents $\cent(\sigma_\alpha)$ in bits as well the tight upper bound of $T_{W_\gamma}^{\chi(\alpha)}$, which happen to coincide. }
\label{fig:werner_upper_bound}
\end{figure}

This optimization procedure is not limited to witnesses, but rather, any measurement. It can be used to identify the presence of negative conditional entropy, and therefore also entanglement. It is known that $\left(|\phi^+\rangle \langle \phi^+|\right)^{T_A}$ where $(.)^{T_A}$ is the partial transpose with respect to the first subsystem, is an entanglement witness \cite{partial_transpose}. Consider the operator $E = \left(0.3 |\phi^+\rangle \langle \phi^+| + 0.7 I/4 \right)^{T_A}$. It is a Hermitian operator, but not a witness as it is positive definite. We find that the tight upper bound of $T_E^{0.05}$ is approximately  $-0.4541$, thus immediately certifying negative conditional entropy (and therefore also entanglement) for any state in which the expectation value of $E$ is $0.05$. We now give an alternate construction of a witness for conditional entropy.

\section{Geometric witness for negative conditional entropy}
\label{sec:geo}
\subsection{Construction of the witness}
The witness constructed in Theorem \ref{witness} can only be constructed for states corresponding to full-rank matrices. In this section, we discuss a witness construction that is valid for any state. We show that an operator constructed in the same fashion as one that witnesses entangled states from the class of separable states(\cite{bertlmann}, \cite{ref_14}, \cite{ref_25}) in fact acts as a witness for a state outside any convex and compact set, and thus also for $\Free$. In the proof that follows, $\|.\|$ refers to the Frobenius norm and $\langle A, B \rangle = \Tr(A^\dagger B)$. Also, arguments to $\langle A, B \rangle$ are Hermitian, hence $\langle A, B \rangle$ is always real. We have $\|A\| =  \sqrt{\langle A, A \rangle}$. 

\begin{theorem}
  \label{numerical_witness_theorem}
  A witness operator $W^g$ separating state $\rho_s$ from any convex and compact set $S \subseteq \All$ is given by,
  \[ W^g = \frac{\Tr(\sigma_c \rho_s - \sigma_c^2)I + \sigma_c - \rho_s}{\sqrt {\Tr(\sigma_c - \rho_s)^2}}, \]
  where 
  \[\sigma_c = \argmin_{\sigma} || \rho_s - \sigma|| \text{ }\forall \text{ } \sigma \in \text{S}.\]
Since $\Free$ is convex and compact, it follows that $W^g$ acts as a witness for a state outside $\Free$.
\end{theorem}\begin{proof}
  \label{numerical_witness_proof}
  To find a witness operator $W^g$ separating a state $\rho_s$ from $S$, given that $\sigma_c$ is the closest state in $S$ to $\rho_s$, we proceed as follows: We consider a function $\func(\chi)$ = $\langle \chi - \sigma_c, \sigma_c - \rho_s \rangle$ defined for all density matrices $\chi$. We show that (i) $\forall \; \sigma \in S, \; \func(\sigma) \geq 0$ and (ii) $\exists \; \rho \notin S$ such that $\func(\rho) < 0$. We then find the $W^g$ such that $\Tr(W^g \chi) = \func(\chi) \; \forall  \; \chi$, and thus find the witness $W^g$.

(i) To prove that the $\func(\sigma)$ is non-negative for all $\sigma \in S$, we assume otherwise:

\begin{equation}
  \label{contradiction}
  \exists \sigma' \in S \;|\; \langle \sigma' - \sigma_c, \sigma_c - \rho_s \rangle < 0.
\end{equation}
Within this assumption, we consider two cases. Intuitively, these two cases correspond to having an (i) obtuse or right angle
between $\sigma_c - \sigma'$ and $\rho_s - \sigma'$ and having an (ii) acute angle between them as shown in Fig \ref{fig:obtuse} and Fig \ref{fig:acute}.

\emph{Case 1:} $\langle \rho_s - \sigma', \sigma_c - \sigma' \rangle \leq 0$. Here, we show that $\sigma'$ is closer to $\rho_s$ than $\sigma_c$, hence leading
to a contradiction. Consider the expression,
\begin{equation}
  \langle \rho_s - \sigma_c, \rho_s - \sigma_c \rangle - 
\langle \rho_s - \sigma', \rho_s - \sigma' \rangle,
\end{equation}
which can be rewritten as follows:
\begin{equation}
    \label{case_1_manip}
    \begin{aligned}
        & \langle \rho_s - \sigma_c, \rho_s - \sigma_c \rangle - \langle \rho_s - \sigma', \rho_s - \sigma' \rangle                                                        \\
    = & \langle \rho_s - \sigma_c, \rho_s - \sigma_c \rangle - \langle \rho_s - \sigma' + \sigma_c - \sigma_c, \rho_s - \sigma' \rangle                                      \\
    = & \langle \rho_s - \sigma_c, \rho_s - \sigma_c \rangle - \langle \rho_s - \sigma_c, \rho_s - \sigma' \rangle - \langle \sigma_c - \sigma', \rho_s - \sigma' \rangle       \\
    = & \langle \rho_s - \sigma_c, \rho_s - \sigma_c - \rho_s + \sigma'\rangle - \langle \sigma_c - \sigma', \rho_s - \sigma' \rangle                                       \\
    = & \langle \rho_s - \sigma_c, \sigma' - \sigma_c \rangle - \langle \sigma_c - \sigma', \rho_s - \sigma'\rangle                                                     \\
    = & \langle \rho_s - \sigma_c + \sigma' - \sigma', \sigma' - \sigma_c \rangle - \langle \sigma_c - \sigma', \rho_s - \sigma'\rangle                                   \\
    = & \langle \rho_s - \sigma', \sigma' - \sigma_c \rangle + \langle \sigma' - \sigma_c, \sigma' - \sigma_c \rangle - \langle \sigma_c - \sigma', \rho_s - \sigma' \rangle \\
    = & \langle \sigma' - \sigma_c, \sigma' - \sigma_c \rangle - 2\langle \rho_s - \sigma', \sigma_c - \sigma' \rangle.
    \end{aligned}
\end{equation}
We know that $\sigma' \neq \sigma_c$, as that would make the initial 
assumption \ref{contradiction} untrue. Hence, the first term is always 
positive by definition of inner product and the second term is non-positive by the assumption in 
Case 1.

Hence, we have $\langle \rho_s - \sigma_c, \rho_s - \sigma_c \rangle - \langle \rho_s - \sigma', \rho_s - \sigma' \rangle > 0$ or $\sqrt{\langle \rho_s - \sigma_c, \rho_s - \sigma_c \rangle} > \sqrt{\langle \rho_s - \sigma', \rho_s - \sigma' \rangle}$. Therefore, $||\rho_s - \sigma_c || > ||\rho_s - \sigma' ||$, which is a contradiction.

\emph{Case 2:}  $\langle \rho_s - \sigma', \sigma_c - \sigma' \rangle > 0$.
Here, we show that there will exist a point $\sigma'' \in S$ such that $\langle \rho_s - \sigma'', \sigma_c - \sigma'' \rangle = 0$. Thus, Case 1 is satisfied and following the argument present there, $\sigma''$ will be closer to $\rho_s$ than $\sigma_c$, and a contradiction is attained.

Since, $\sigma' \in S$ and $\sigma_c \in S$, by the convexity of the
set $S$, we have:
\begin{equation}
  \forall \, \lambda \in [0,1],\; \lambda \sigma_c + (1 - \lambda) \sigma' \in S.
\end{equation}
Consider $\lambda'  = 
\frac{\langle \sigma' - \rho_s, \sigma' - \sigma_c \rangle}{\langle \sigma' - \sigma_c, \sigma' - \sigma_c \rangle}$.
By the assumption in Case 2, the numerator is positive, and by the definition of inner product,
the denominator is positive. Therefore $\lambda' > 0$. 
We also have $\lambda'$ = $\frac{\langle \sigma' - \sigma_c + \sigma_c - \rho_s, \sigma' - \sigma_c \rangle}{\langle \sigma' - \sigma_c, \sigma' - \sigma_c \rangle}$ =  $1 - \frac{\langle \sigma' - \sigma_c, \rho_s - \sigma_c \rangle}{\langle \sigma' - \sigma_c, \sigma' - \sigma_c \rangle}$, and
by the initial contradiction assumption \ref{contradiction}, the second term is positive, hence $\lambda' < 1$. 

Therefore, $\lambda' \in (0,1)$. 
Consider a state $\sigma'' = \lambda' \sigma_c + (1 - \lambda') \sigma'$. Notice that this point belongs to the set $S$ and
$\langle \rho_s - \sigma'', \sigma_c - \sigma'' \rangle = 0$, therefore satisfying Case 1.
Thus both cases lead to a contradiction rendering the initial assumption \ref{contradiction} untrue. Thus, all points $\sigma$ in $S$ satisfy $\langle \sigma - \sigma_c, \sigma_c - \rho_s \rangle \geq 0$.

(ii) To show that $\exists \rho \notin S \;|\; \func(\rho) < 0$, we consider the point $\rho_s$ itself, which we wish to separate. It is easy to see that $\func(\rho_s)$ gives a negative value: 
\begin{equation}
    \langle \rho_s - \sigma_c, \sigma_c - \rho_s \rangle = -||\rho_s - \sigma_c ||^2 < 0.
\end{equation}
We find the witness operator $W^g$ such that $\Tr(W^g \chi) = \langle \chi - \sigma_c, \sigma_c - \rho_s \rangle = \Tr((\chi - \sigma_c)(\sigma_c - \rho_s))$. On solving the above,  $W^g = \Tr(\sigma_c \rho_s - \sigma_c^2)I + \sigma_c - \rho_s$. On normalising, we get $W^g$ as mentioned in Theorem \ref{numerical_witness_theorem}. 
\end{proof}

% \begin{figure}
%   \includegraphics[width=\linewidth]{separating_hyperplane.png}
%   \caption{The Separating Hyperplane}
%   \label{fig:Separating Hyperplane}
% \end{figure}

\begin{figure}
  \includegraphics[width=\linewidth]{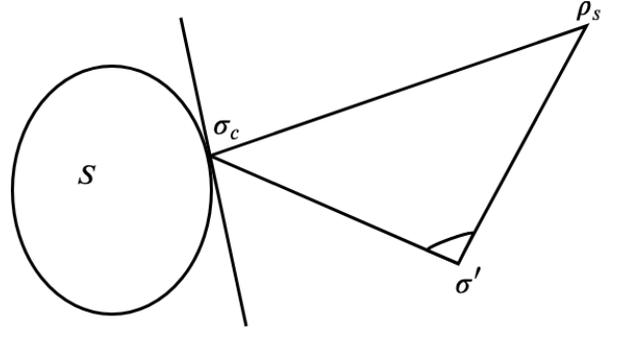}
  \caption{Case 1: Corresponds to angle at $\sigma'$ being obtuse.}
  \label{fig:obtuse}
\end{figure}
\begin{figure}
  \includegraphics[width=\linewidth]{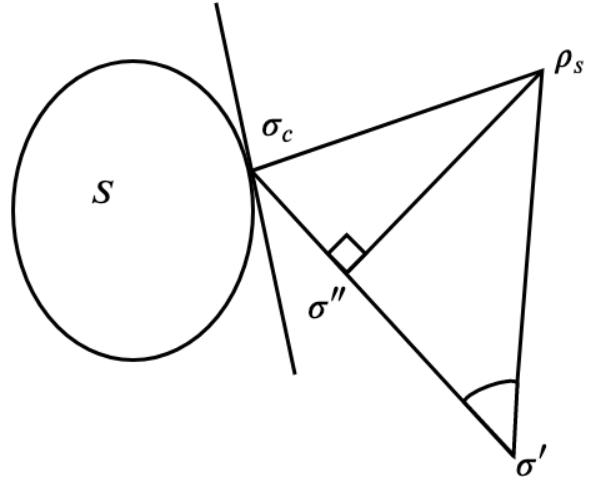}
  \caption{Case 2: Corresponds to the angle at $\sigma'$ being acute.}
  \label{fig:acute}
\end{figure}

\begin{figure}
\centering
  \includegraphics[width=0.9\linewidth]{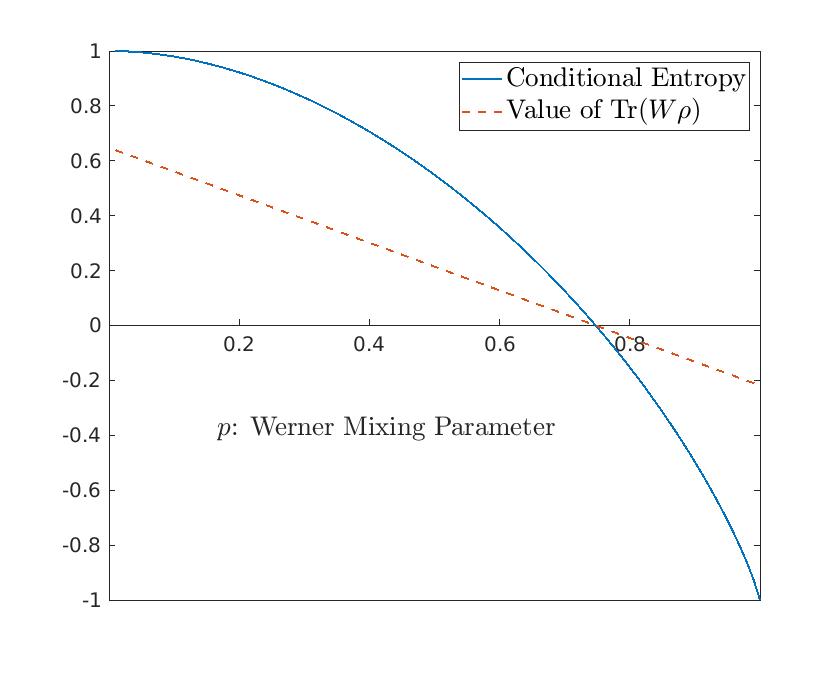}
  \caption{Variation of conditional entropy and
  the value of $\Tr(W\rho)$ with the mixing parameter $p$. The $x-$axis is $p$, the Werner mixing parameter, and refers to two dimensional states of the form $\rho =  p | \phi^+ \rangle \langle \phi^+ | +  (1 - p)\frac{I}{d^2}$. The $y-$axis represents different quantities for the red (dashed) line and the blue (solid) line, placed together for comparison. For the red (dashed) line, it represents the value of $\Tr(W \rho)$ where $W$ is the geometric witness calculated in Eq.\ref{example_witness} and for the blue (solid) line, it represents the conditional entropy in bits.}
  \label{fig:numerical_witness_plot}
\end{figure}

Now, it remains that we find $\sigma_c$, the closest state in $\Free$ to $\rho_s$. Note that trace distance $|| \rho_s - \sigma ||$ is a convex function in $\sigma$ for a fixed $\rho_s$, 
and it needs to be optimised over all $\sigma \in \Free$, a convex set. Also, unlike the entanglement problem, the conditional entropy expression can be evaluated for any given dimension at any state, and membership in $\Free$ can be tested. Thus, this is a convex optimization problem and there exist several solvers to tackle this. We use MATLAB's CVX solver \cite{cvx_1, cvx_2}, a package for specifying and solving convex programs and a few quantum function libraries \cite{cvxquad, quantinf} to find the closest state in $\Free$ for several 2-qubit and 3-qubit states.

\subsection{Example: Werner class of states}
In the $2 \times 2$ dimensional space, consider the Bell state $|\phi^+ \rangle \langle \phi^+|$. We find that the closest state in $\Free$ to $|\phi^+ \rangle \langle \phi^+|$ is the following:
\[\begin{bmatrix}
  a & 0 & 0 & c \\
  0 & b & 0 & 0 \\
  0 & 0 & b & 0 \\
  c & 0 & 0 & a \\
\end{bmatrix}\]
where $ a \approx 0.4369, b \approx 0.0631, c \approx 0.3738 $.
We note that this is the state where the Werner line $p |\phi^+\rangle \langle \phi^+ | + (1-p)\frac{I}{4}$ intersects $\Free$. Therefore, this state at $p \approx 0.7476$ is the closest state in $\Free$ for all states with negative conditional entropy that lie on the the given Werner line. Hence the witness created using this state witnesses all negative conditional entropy states of the form $p |\phi^+\rangle \langle \phi^+ | + (1-p)\frac{I}{4}$ where $0 \leq p \leq 1$.

Applying Theorem \ref{numerical_witness_theorem}, we find the witness to this state is:
\begin{equation}
\label{example_witness}
W = \begin{bmatrix}
  a & 0 & 0 & c \\
  0 & b & 0 & 0 \\
  0 & 0 & b & 0 \\
  c & 0 & 0 & a \\
\end{bmatrix}
\end{equation}
where $a \approx 0.3588$, $b \approx 0.9361$ and $c  \approx -0.5774$. To check that this operator indeed acts as a witness, consider the state $\rho_w = 0.75 |\phi^+\rangle \langle \phi^+| +  0.25 \frac{I}{4}$ which lies on the Werner line just outside $\Free$:
\begin{equation}
\Tr(W \rho_w) \approx -0.0021
\end{equation}
which is a negative value as it is outside the class, as expected. Fig \ref{fig:numerical_witness_plot} shows the variation of $\Tr(W\rho)$ with mixing parameter $p$, as well as the corresponding conditional entropy. 

The witness operator $W$ can readily be decomposed in the form of Pauli matrices as follows and be implemented in the laboratory for measurement:
\[W = \frac{a+b}{2} I \otimes I + \frac{a-b}{2} Z \otimes Z + \frac{c}{2} X \otimes X - \frac{c}{2} Y \otimes Y\]

\subsection{Proof of weak optimality}
An entanglement witness $A$ is ``weakly optimal" if there exists a separable state $\sigma_s$ such that $\langle \sigma_s, A \rangle = 0$ \cite{weakly_optimal, ref_25}. Similarly, we define weak optimality for witnesses of negative conditional entropy: A witness of negative conditional entropy $W$ is ``weakly optimal" if there exists a state $\sigma_o$ such that (i) $\cent(\sigma_o) \geq 0$ and (ii) $\langle \sigma_o, W \rangle = 0$. To show that $W^g$ is weakly optimal, we use the state $\sigma_c$.  Using the definitions in the proof of Theorem \ref{numerical_witness_theorem}, we have (i) $\cent(\sigma_c) \geq 0$ and (ii):
\[
    \langle \sigma_c, W^g \rangle
    = \Tr(\sigma_c W^g)
    = \func(\sigma_c)
    = \langle \sigma_c - \sigma_c, \sigma_c - \rho_s \rangle
    = 0.
\]
Hence, $W^g$ is a weakly optimal witness for negative conditional entropy. Such a witness is also known as a tangent functional \cite{ref_14, ref_25}. We now discuss some applications of witnessing negative conditional entropy.

\section{Applications}
\label{sec:app}
\subsection{Quantum state merging}
When a quantum state $\rho_{AB}$ is shared between $A$ and $B$, we ask: How much quantum communication is needed  from $A$ to transfer the full information of the state to $B$? This process of merging the missing information from one system to another is called state merging \cite{statemerging1,statemerging2}. The quantification of this process is given by the conditional entropy $\cent(\rho_{AB})$ (if it is from $A$ to $B$) of the system. If $\cent(\rho_{AB})$  is positive, it means that the sender needs to communicate this number of quantum bits to the receiver and if it is $0$, there is no need for such communication. However, when  it is negative, the sender and receiver gain  corresponding potential for  future quantum communication between them even after the merging process. The witness provided by us detects such states, and is therefore useful for the several applications of the quantum state merging primitive detailed in \cite{statemerging2}. 

\subsection{Superdense coding}
If two parties, Alice and Bob share an entangled state, then by sending a qudit to Bob, Alice may communicate more information than possible with a classical dit, which is known as superdense coding. The superdense coding capacity for a mixed state $\rho_{AB}$ in %$D(\sent_d \otimes \sent_d)$ 
is given by \cite{sdc2} \\
\begin{equation}
    C_{AB} = \max\{\log_2d, \log_2d + \ent(\rho_B) - \ent(\rho_{AB})\}.
\end{equation}
Here, $ C_{AB} $ is the amount of classical information that can be sent from  A to B using $\rho_{AB}$ as a quantum resource. Note that when $\cent(\rho_{AB})$ is negative we  can use the shared state to transfer classical communication greater than the classical limit of $\log_2 d$ bits. Therefore, states with negative $\cent(\sigma_{AB})$ give quantum advantage while performing superdense coding and our result is useful in detecting such states. In addition, dense coding capacity also obeys the exclusion principle in the case of multipartite systems in which negative conditional entropy plays an important role \cite{sdc3}.

\subsection{Distributed private randomness distillation} The inherent probabilistic nature of quantum measurement outcomes provides a natural arena for generation of randomness  \cite{randomness_extractors}. The generated randomness is considered private if it is completely unpredictable in advance \cite{private_rand}. For example, measurement of the $|+ \rangle$ state in the computational basis is private as a pure state is not correlated with any other state, including one with a  possible Eavesdropper. However, this may not always be the case. In \cite{randomness_extractors}, the problem of extracting randomness from a state where part of the state is in the possession of the Eavesdropper is considered. In \cite{private_rand}, a distributed setting of the above problem is considered: if two parties, Alice and Bob trust each other and share a state $\rho_{AB}$ and Eve possesses the purification of the state, then what is the maximum rate of randomness, that Alice and Bob ($R_A$ and $R_B$ respectively) can extract private from Eve?

In the case that Alice and Bob cannot communicate with each other and can only perform local unitary operations, while each having access to free noise (maximally mixed states), the bound turns out to be: $R_A \leq \log|A| - \cent(\rho_{AB}), R_B \leq \log|B| - \ent_{B|A}(\rho_{AB})$ and $R_A + R_B \leq R_G$ where $R_G = \log|AB| - \ent(\rho_{AB})$ and $|X|$ is the dimensionality of the quantum system $X$.

%j Our result shows that we can detect states with a negative conditional von Neumann entropy, and this in principle helps us characterize states $\rho_{AB}$ that have the upper bound of  $R_A > \log|A|$ and the upper bound of $R_B > \log|B|$, respectively. 

Alice and Bob can use the witness provided by us to obtain an upper bound on $\cent(\rho_{AB})$ and $\cent(\rho_{AB})$. For a two qubit system, we have provided a local decomposition of the witness to accomplish this (\ref{decomposition}). Note that classical communication is necessary for Alice and Bob to estimate these upper bounds. 

\subsection{One-way entanglement distillation}
Entanglement distillation is the process of transformation of $N$ copies of an arbitrary entangled state $\rho_{AB}$ into some number of approximately pure EPR pairs, using only local operations and classical communication (LOCC) at optimal rates \cite{distillation}. In one-way entanglement distillation, the goal is also to transform a given shared state into maximally entangled states, albeit in an asymmetric setting as described in \cite{one_way_entangle}. Here, the one-way (or forward) entanglement capacity of $\rho_{AB}$, which is the maximal achievable rate of entanglement distillation, is given by the Hashing inequality \cite{one_way_entangle}: $\hash(\rho) \geq - \cent(\rho)$. The witness provided by us detects states with a positive one-way entanglement capacity, and the subsequent optimization procedure can be utilized to provide a tight lower bound on the same. 

\subsection{Coherent information}
Consider a quantum system $A$ and a reference system $P$. Initially, the quantum system
$A$ and the reference system $P$ are in a pure entangled state,
denoted by $|\psi\rangle _{PA}$.  The density matrix $\rho_A$ of the initial system $A$ is given by $\rho_A=\Tr_P(|\psi\rangle _{PA} \langle \psi |)$. Consider an unknown quantum channel $N$ through which the state $\rho_A$ of the system $A$ is transmitted while the system $P$ remains isolated from the environment and does not get affected by the channel. The combined output state $\rho_{PB}$ is given by $\rho_{PB}= (I_P \otimes N_A)(|\psi\rangle _{PA} \langle \psi |)$ as shown in Fig \ref{fig:qchannel}. 

As a result of this transmission, both the system $A$ and the entanglement between the system $A$ and $P$ get affected. The quantum capacity $\qcap(N)$ of the channel is associated with a correlation measure known as \textit{quantum coherent information} \cite{coherent}. 
The entropy exchange between the initial system $PA$ and the output system $PB$ is a measure of the extent to which the initially pure state $\ket \psi_{PA}$ becomes mixed as a result of the channel, and is defined as $\ent_{E}(\rho_A; N(\rho_{A}))= \ent(\rho_{PB})$. A lower bound for the quantum capacity $\qcap(N)$ is given by the quantum coherent information which is expressed as $\ent(N(\rho_A))- \ent_{E}(\rho_A; N(\rho_{A}))= \ent(\rho_B)- \ent(\rho_{PB})$ \cite{capacity}. The expression $\ent(\rho_B)- \ent(\rho_{PB})$ is the negative of the conditional entropy $\ent_{\text{P}|\text{B}}(\rho_{PB})$. Thus, if we have a negative value of $\ent_{\text{P}|\text{B}}(\rho_{PB})$, then we have a positive value of quantum coherent information, and therefore of $\qcap(N)$. Our witness operator can detect whether the state $\rho_{PB}$ has negative conditional entropy. This implies that we can detect whether the unknown quantum channel $N$ is useful in transmitting quantum information.

\begin{figure}[h]
\centering
\includegraphics[scale=.5]{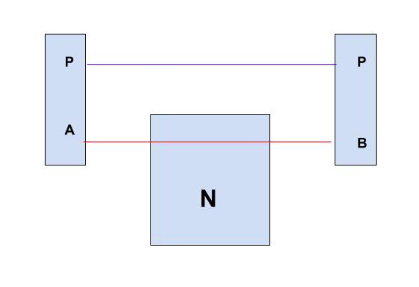}
\caption{ The dynamics of the system $A$ when sent through the quantum channel $N$  with $P$ as the reference system. $B$ is the output of the system $A$ after the channel transmission.}
\label{fig:qchannel}
\end{figure}
\subsection{Uncertainty relations}
The entropic uncertainty relation \cite{ur} for two incompatible measurements $X$ and $Y$ performed on a quantum state is given by  \cite{complementary, generalized_uncertainty}: 
\begin{equation}
    \label{uncertainty_1}
    \sent(X) + \sent(Y) \geq \log_2{\frac{1}{c}}.
\end{equation}
Here, $\sent(X)$ and $\sent(Y)$ denote the Shannon entropy of the outcomes of measurement $X$ and $Y$, respectively when performed on any quantum state, and $c=\max_{i,j}|\langle x_i| y_j \rangle|^2$ where $|x_i\rangle,|y_j\rangle$ are eigenvectors of $X$ and $Y$, respectively.

The operational interpretation of Eq.~\ref{uncertainty_1} is as follows: Alice and Bob agree on two measurements $X$ and $Y$. Bob prepares a state and sends it to Alice who measures, according to her choice, $X$ or $Y$ and reveals her measurement choice ($X$ or $Y$). Bob then guesses her outcome. Bob's aim is to prepare the initial state such that he reduces the uncertainty about the outcome revealed by Alice. By Eq.~\ref{uncertainty_1}, whatever state Bob prepares, $\sent(X) + \sent(Y)$ will be at least $\log_2 \frac{1}{c}$.

However, note that if Bob has access to quantum memory, he can prepare a state that is maximally entangled with the memory and send it to Alice, then apply the same measurement that Alice picks on his own memory, and can guess Alice's outcome with full certainty. Hence, the above inequality does not hold in the presence of quantum memory. A new uncertainty relation taking into consideration quantum memory \cite{ur1} is the following:
\begin{equation}
    \label{uncertainty_2}
    \osent(X|B) + \osent(Y|B) \geq \log_2 \frac{1}{c} + \cent(\rho_{AB}),
\end{equation}
where $\rho_{AB}$ is the combined state of Bob's quantum memory and the subsystem he sends to Alice. $\osent(X|B)$ (or $\osent(Y|B)$) quantifies the uncertainty about the outcome of measurement $X$ (or $Y$), given information stored in a quantum memory $B$ \cite{ur1}.

Here we notice that whenever $\cent(\rho_{AB})$ is negative, by Eq.~\ref{uncertainty_2}, Bob's prepared state beats the original uncertainty bound of Eq.~\ref{uncertainty_1}. In other words, by our characterization of $\Free$, we make it possible to experimentally witness states that have the ability to utilize the available quantum memory in overcoming the original uncertainty bound. Among these states, are the states with $\cent(\rho_{AB}) \leq \log_2c$, for which the uncertainty bound becomes trivial:  $\osent(X|B) + \osent(Y|B) \geq 0$. Fig.~\ref{fig:quantum_memory} depicts this classification of states for a given measurement set $\{X, Y\}$.

\begin{figure}
    \includegraphics[width=0.6\linewidth]{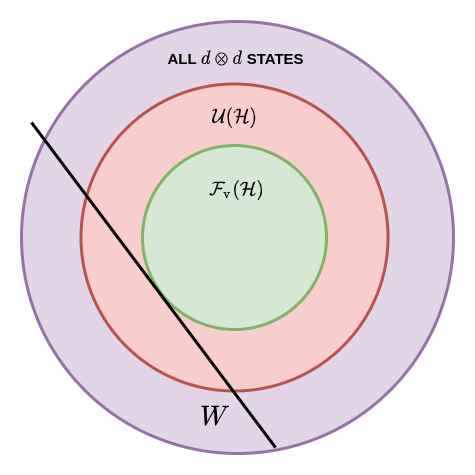}
    \caption{The inner set represents $\Free$ --- the set of states that give no advantage with quantum memory over the original scenario. Beyond it, labelled $\cal U (\cal H)$, are the set of states that do give an advantage, but still have a degree of uncertainty. And beyond those (outermost) are the states with $\cent(\rho_{AB}) \leq \log_2c$, for which the uncertainty bound is trivial.}
\label{fig:quantum_memory}
\end{figure}

\section{Conclusion}
\label{sec:conclusion}
In the present work, we prove that the states having non-negative conditional entropy form a convex and compact set. This allows for the construction of witnesses to detect states with negative
conditional entropy. We provide two constructions of such witnesses and for one of them, we prove that the expectation value of the witness in a state is an upper bound to the conditional entropy of the state. We pose the problem of obtaining a tight upper bound to the set of conditional entropies of states in which an operator gives the same expectation value as a convex optimization problem, and discuss several interesting results. We believe that further insight on conditional entropy as a resource distinct from entanglement can be gained by studying it in a resource-theoretic framework. In such a study, the set of states with non-negative conditional entropy characterized in this work would be the set of free states.

% {Conditional von Neumann entropy plays an important role in quantum information processing tasks like superdense coding, state merging, uncertainty relations with quantum memory. one way entanglement distillation and private randomness extraction. However, a proper characterization of the resource was lacking unlike entanglement. Even within the set of entangled states not all states prove to be useful in certain information processing tasks.}
% \par In the present paper, we found that the states having non-negative conditional entropy form a convex and compact set. This allows for the construction of Hermitian operators to witness states with negative conditional entropy. 
% \par We have provided a prescription for the construction of the witness and exemplified its action on certain states. This underscores the utility of our prescription pertaining to informational tasks. Decomposition of the witness in terms of local observables further facilitates practical implementation. 

% \par The present work also raises some pertinent questions for future research. An extension of the work in multipartite systems is a significant area to probe upon. Further developing the resource theory of conditional von Neumann entropy, including characterizing the class of free operations for $cal{F}_{mathrm v} (cal H)$ is another immediate area of research.\\ 
\begin{acknowledgments}
We acknowledge Prof. John Watrous, Dr. Luca Innocenti and Mr. Siddharth Bhat for having useful discussions at various stages of completion of this manuscript. We also thank the anonymous referees for their valuable suggestions, which were immensely useful in improving out work. N.G. would like to acknowledge support from the Research Initiation Grant of BITS-Pilani, Hyderabad vide letter no. BITS/GAU/RIG/2019/H0680 dated 22nd April, 2019. 
\end{acknowledgments}

\appendix*
\bibliographystyle{apsrev4-1}

\begin{thebibliography}{42}%
\makeatletter
\providecommand \@ifxundefined [1]{%
 \@ifx{#1\undefined}
}%
\providecommand \@ifnum [1]{%
 \ifnum #1\expandafter \@firstoftwo
 \else \expandafter \@secondoftwo
 \fi
}%
\providecommand \@ifx [1]{%
 \ifx #1\expandafter \@firstoftwo
 \else \expandafter \@secondoftwo
 \fi
}%
\providecommand \natexlab [1]{#1}%
\providecommand \enquote  [1]{``#1''}%
\providecommand \bibnamefont  [1]{#1}%
\providecommand \bibfnamefont [1]{#1}%
\providecommand \citenamefont [1]{#1}%
\providecommand \href@noop [0]{\@secondoftwo}%
\providecommand \href [0]{\begingroup \@sanitize@url \@href}%
\providecommand \@href[1]{\@@startlink{#1}\@@href}%
\providecommand \@@href[1]{\endgroup#1\@@endlink}%
\providecommand \@sanitize@url [0]{\catcode `\\12\catcode `\$12\catcode
  `\&12\catcode `\#12\catcode `\^12\catcode `\_12\catcode `\%12\relax}%
\providecommand \@@startlink[1]{}%
\providecommand \@@endlink[0]{}%
\providecommand \url  [0]{\begingroup\@sanitize@url \@url }%
\providecommand \@url [1]{\endgroup\@href {#1}{\urlprefix }}%
\providecommand \urlprefix  [0]{URL }%
\providecommand \Eprint [0]{\href }%
\providecommand \doibase [0]{http://dx.doi.org/}%
\providecommand \selectlanguage [0]{\@gobble}%
\providecommand \bibinfo  [0]{\@secondoftwo}%
\providecommand \bibfield  [0]{\@secondoftwo}%
\providecommand \translation [1]{[#1]}%
\providecommand \BibitemOpen [0]{}%
\providecommand \bibitemStop [0]{}%
\providecommand \bibitemNoStop [0]{.\EOS\space}%
\providecommand \EOS [0]{\spacefactor3000\relax}%
\providecommand \BibitemShut  [1]{\csname bibitem#1\endcsname}%
\let\auto@bib@innerbib\@empty
%</preamble>
\bibitem [{\citenamefont {Horodecki}\ \emph {et~al.}(2009)\citenamefont
  {Horodecki}, \citenamefont {Horodecki}, \citenamefont {Horodecki},\ and\
  \citenamefont {Horodecki}}]{enthorodecki}%
  \BibitemOpen
  \bibfield  {author} {\bibinfo {author} {\bibfnamefont {R.}~\bibnamefont
  {Horodecki}}, \bibinfo {author} {\bibfnamefont {P.}~\bibnamefont
  {Horodecki}}, \bibinfo {author} {\bibfnamefont {M.}~\bibnamefont
  {Horodecki}}, \ and\ \bibinfo {author} {\bibfnamefont {K.}~\bibnamefont
  {Horodecki}},\ }\href {\doibase 10.1103/RevModPhys.81.865} {\bibfield
  {journal} {\bibinfo  {journal} {Rev. Mod. Phys.}\ }\textbf {\bibinfo {volume}
  {81}},\ \bibinfo {pages} {865} (\bibinfo {year} {2009})}\BibitemShut
  {NoStop}%
\bibitem [{\citenamefont {Ganguly}\ \emph {et~al.}(2011)\citenamefont
  {Ganguly}, \citenamefont {Adhikari}, \citenamefont {Majumdar},\ and\
  \citenamefont {Chatterjee}}]{telwit}%
  \BibitemOpen
  \bibfield  {author} {\bibinfo {author} {\bibfnamefont {N.}~\bibnamefont
  {Ganguly}}, \bibinfo {author} {\bibfnamefont {S.}~\bibnamefont {Adhikari}},
  \bibinfo {author} {\bibfnamefont {A.~S.}\ \bibnamefont {Majumdar}}, \ and\
  \bibinfo {author} {\bibfnamefont {J.}~\bibnamefont {Chatterjee}},\ }\href
  {\doibase 10.1103/PhysRevLett.107.270501} {\bibfield  {journal} {\bibinfo
  {journal} {Phys. Rev. Lett.}\ }\textbf {\bibinfo {volume} {107}},\ \bibinfo
  {pages} {270501} (\bibinfo {year} {2011})}\BibitemShut {NoStop}%
\bibitem [{\citenamefont {Cerf}\ and\ \citenamefont
  {Adami}(1997)}]{negentropy}%
  \BibitemOpen
  \bibfield  {author} {\bibinfo {author} {\bibfnamefont {N.~J.}\ \bibnamefont
  {Cerf}}\ and\ \bibinfo {author} {\bibfnamefont {C.}~\bibnamefont {Adami}},\
  }\href {\doibase 10.1103/PhysRevLett.79.5194} {\bibfield  {journal} {\bibinfo
   {journal} {Phys. Rev. Lett.}\ }\textbf {\bibinfo {volume} {79}},\ \bibinfo
  {pages} {5194} (\bibinfo {year} {1997})}\BibitemShut {NoStop}%
\bibitem [{\citenamefont {Abe}\ and\ \citenamefont
  {Rajagopal}(2001)}]{non-additivity}%
  \BibitemOpen
  \bibfield  {author} {\bibinfo {author} {\bibfnamefont {S.}~\bibnamefont
  {Abe}}\ and\ \bibinfo {author} {\bibfnamefont {A.}~\bibnamefont
  {Rajagopal}},\ }\href {\doibase
  https://doi.org/10.1016/S0378-4371(00)00476-3} {\bibfield  {journal}
  {\bibinfo  {journal} {Physica A: Statistical Mechanics and its Applications}\
  }\textbf {\bibinfo {volume} {289}},\ \bibinfo {pages} {157} (\bibinfo {year}
  {2001})}\BibitemShut {NoStop}%
\bibitem [{\citenamefont {Horodecki}\ \emph {et~al.}(2005)\citenamefont
  {Horodecki}, \citenamefont {Oppenheim},\ and\ \citenamefont
  {Winter}}]{statemerging1}%
  \BibitemOpen
  \bibfield  {author} {\bibinfo {author} {\bibfnamefont {M.}~\bibnamefont
  {Horodecki}}, \bibinfo {author} {\bibfnamefont {J.}~\bibnamefont
  {Oppenheim}}, \ and\ \bibinfo {author} {\bibfnamefont {A.}~\bibnamefont
  {Winter}},\ }\href {\doibase 10.1038/nature03909} {\bibfield  {journal}
  {\bibinfo  {journal} {Nature}\ }\textbf {\bibinfo {volume} {436}},\ \bibinfo
  {pages} {673} (\bibinfo {year} {2005})}\BibitemShut {NoStop}%
\bibitem [{\citenamefont {Horodecki}\ \emph {et~al.}(2007)\citenamefont
  {Horodecki}, \citenamefont {Oppenheim},\ and\ \citenamefont
  {Winter}}]{statemerging2}%
  \BibitemOpen
  \bibfield  {author} {\bibinfo {author} {\bibfnamefont {M.}~\bibnamefont
  {Horodecki}}, \bibinfo {author} {\bibfnamefont {J.}~\bibnamefont
  {Oppenheim}}, \ and\ \bibinfo {author} {\bibfnamefont {A.}~\bibnamefont
  {Winter}},\ }\href {\doibase 10.1007/s00220-006-0118-x} {\bibfield  {journal}
  {\bibinfo  {journal} {Communications in Mathematical Physics}\ }\textbf
  {\bibinfo {volume} {269}},\ \bibinfo {pages} {107} (\bibinfo {year}
  {2007})}\BibitemShut {NoStop}%
\bibitem [{\citenamefont {Bennett}\ and\ \citenamefont {Wiesner}(1992)}]{sdc1}%
  \BibitemOpen
  \bibfield  {author} {\bibinfo {author} {\bibfnamefont {C.~H.}\ \bibnamefont
  {Bennett}}\ and\ \bibinfo {author} {\bibfnamefont {S.~J.}\ \bibnamefont
  {Wiesner}},\ }\href {\doibase 10.1103/PhysRevLett.69.2881} {\bibfield
  {journal} {\bibinfo  {journal} {Phys. Rev. Lett.}\ }\textbf {\bibinfo
  {volume} {69}},\ \bibinfo {pages} {2881} (\bibinfo {year}
  {1992})}\BibitemShut {NoStop}%
\bibitem [{\citenamefont {Bru\ss{}}\ \emph {et~al.}(2004)\citenamefont
  {Bru\ss{}}, \citenamefont {D'Ariano}, \citenamefont {Lewenstein},
  \citenamefont {Macchiavello}, \citenamefont {Sen(De)},\ and\ \citenamefont
  {Sen}}]{sdc2}%
  \BibitemOpen
  \bibfield  {author} {\bibinfo {author} {\bibfnamefont {D.}~\bibnamefont
  {Bru\ss{}}}, \bibinfo {author} {\bibfnamefont {G.~M.}\ \bibnamefont
  {D'Ariano}}, \bibinfo {author} {\bibfnamefont {M.}~\bibnamefont
  {Lewenstein}}, \bibinfo {author} {\bibfnamefont {C.}~\bibnamefont
  {Macchiavello}}, \bibinfo {author} {\bibfnamefont {A.}~\bibnamefont
  {Sen(De)}}, \ and\ \bibinfo {author} {\bibfnamefont {U.}~\bibnamefont
  {Sen}},\ }\href {\doibase 10.1103/PhysRevLett.93.210501} {\bibfield
  {journal} {\bibinfo  {journal} {Phys. Rev. Lett.}\ }\textbf {\bibinfo
  {volume} {93}},\ \bibinfo {pages} {210501} (\bibinfo {year}
  {2004})}\BibitemShut {NoStop}%
\bibitem [{\citenamefont {Prabhu}\ \emph {et~al.}(2013)\citenamefont {Prabhu},
  \citenamefont {Pati}, \citenamefont {Sen(De)},\ and\ \citenamefont
  {Sen}}]{sdc3}%
  \BibitemOpen
  \bibfield  {author} {\bibinfo {author} {\bibfnamefont {R.}~\bibnamefont
  {Prabhu}}, \bibinfo {author} {\bibfnamefont {A.~K.}\ \bibnamefont {Pati}},
  \bibinfo {author} {\bibfnamefont {A.}~\bibnamefont {Sen(De)}}, \ and\
  \bibinfo {author} {\bibfnamefont {U.}~\bibnamefont {Sen}},\ }\href {\doibase
  10.1103/PhysRevA.87.052319} {\bibfield  {journal} {\bibinfo  {journal} {Phys.
  Rev. A}\ }\textbf {\bibinfo {volume} {87}},\ \bibinfo {pages} {052319}
  (\bibinfo {year} {2013})}\BibitemShut {NoStop}%
\bibitem [{\citenamefont {Devetak}\ and\ \citenamefont
  {Winter}(2005)}]{one_way_entangle}%
  \BibitemOpen
  \bibfield  {author} {\bibinfo {author} {\bibfnamefont {I.}~\bibnamefont
  {Devetak}}\ and\ \bibinfo {author} {\bibfnamefont {A.}~\bibnamefont
  {Winter}},\ }\href {\doibase 10.1098/rspa.2004.1372} {\bibfield  {journal}
  {\bibinfo  {journal} {Proceedings of the Royal Society A: Mathematical,
  Physical and Engineering Sciences}\ }\textbf {\bibinfo {volume} {461}},\
  \bibinfo {pages} {207} (\bibinfo {year} {2005})}\BibitemShut {NoStop}%
\bibitem [{\citenamefont {Berta}\ \emph {et~al.}(2010)\citenamefont {Berta},
  \citenamefont {Christandl}, \citenamefont {Colbeck}, \citenamefont {Renes},\
  and\ \citenamefont {Renner}}]{ur1}%
  \BibitemOpen
  \bibfield  {author} {\bibinfo {author} {\bibfnamefont {M.}~\bibnamefont
  {Berta}}, \bibinfo {author} {\bibfnamefont {M.}~\bibnamefont {Christandl}},
  \bibinfo {author} {\bibfnamefont {R.}~\bibnamefont {Colbeck}}, \bibinfo
  {author} {\bibfnamefont {J.~M.}\ \bibnamefont {Renes}}, \ and\ \bibinfo
  {author} {\bibfnamefont {R.}~\bibnamefont {Renner}},\ }\href {\doibase
  10.1038/nphys1734} {\bibfield  {journal} {\bibinfo  {journal} {Nature
  Physics}\ }\textbf {\bibinfo {volume} {6}},\ \bibinfo {pages} {659} (\bibinfo
  {year} {2010})}\BibitemShut {NoStop}%
\bibitem [{\citenamefont {Yang}\ \emph {et~al.}(2019)\citenamefont {Yang},
  \citenamefont {Horodecki},\ and\ \citenamefont {Winter}}]{private_rand}%
  \BibitemOpen
  \bibfield  {author} {\bibinfo {author} {\bibfnamefont {D.}~\bibnamefont
  {Yang}}, \bibinfo {author} {\bibfnamefont {K.}~\bibnamefont {Horodecki}}, \
  and\ \bibinfo {author} {\bibfnamefont {A.}~\bibnamefont {Winter}},\ }\href
  {\doibase 10.1103/PhysRevLett.123.170501} {\bibfield  {journal} {\bibinfo
  {journal} {Phys. Rev. Lett.}\ }\textbf {\bibinfo {volume} {123}},\ \bibinfo
  {pages} {170501} (\bibinfo {year} {2019})}\BibitemShut {NoStop}%
\bibitem [{\citenamefont {Azuma}\ and\ \citenamefont
  {Subramanian}(2018)}]{blackholes_entropy}%
  \BibitemOpen
  \bibfield  {author} {\bibinfo {author} {\bibfnamefont {K.}~\bibnamefont
  {Azuma}}\ and\ \bibinfo {author} {\bibfnamefont {S.}~\bibnamefont
  {Subramanian}},\ }\href@noop {} {\bibfield  {journal} {\bibinfo  {journal}
  {arXiv preprint arXiv:1807.06753}\ } (\bibinfo {year} {2018})}\BibitemShut
  {NoStop}%
\bibitem [{\citenamefont {Azuma}\ and\ \citenamefont
  {Kato}(2020)}]{blackholes_second_law}%
  \BibitemOpen
  \bibfield  {author} {\bibinfo {author} {\bibfnamefont {K.}~\bibnamefont
  {Azuma}}\ and\ \bibinfo {author} {\bibfnamefont {G.}~\bibnamefont {Kato}},\
  }\href@noop {} {\bibfield  {journal} {\bibinfo  {journal} {arXiv preprint
  arXiv:2001.02897}\ } (\bibinfo {year} {2020})}\BibitemShut {NoStop}%
\bibitem [{\citenamefont {del Rio}\ \emph {et~al.}(2011)\citenamefont {del
  Rio}, \citenamefont {{\AA}berg}, \citenamefont {Renner}, \citenamefont
  {Dahlsten},\ and\ \citenamefont {Vedral}}]{thermo}%
  \BibitemOpen
  \bibfield  {author} {\bibinfo {author} {\bibfnamefont {L.}~\bibnamefont {del
  Rio}}, \bibinfo {author} {\bibfnamefont {J.}~\bibnamefont {{\AA}berg}},
  \bibinfo {author} {\bibfnamefont {R.}~\bibnamefont {Renner}}, \bibinfo
  {author} {\bibfnamefont {O.}~\bibnamefont {Dahlsten}}, \ and\ \bibinfo
  {author} {\bibfnamefont {V.}~\bibnamefont {Vedral}},\ }\href {\doibase
  10.1038/nature10123} {\bibfield  {journal} {\bibinfo  {journal} {Nature}\
  }\textbf {\bibinfo {volume} {474}},\ \bibinfo {pages} {61} (\bibinfo {year}
  {2011})}\BibitemShut {NoStop}%
\bibitem [{\citenamefont {Patro}\ \emph {et~al.}(2017)\citenamefont {Patro},
  \citenamefont {Chakrabarty},\ and\ \citenamefont {Ganguly}}]{patro}%
  \BibitemOpen
  \bibfield  {author} {\bibinfo {author} {\bibfnamefont {S.}~\bibnamefont
  {Patro}}, \bibinfo {author} {\bibfnamefont {I.}~\bibnamefont {Chakrabarty}},
  \ and\ \bibinfo {author} {\bibfnamefont {N.}~\bibnamefont {Ganguly}},\ }\href
  {\doibase 10.1103/PhysRevA.96.062102} {\bibfield  {journal} {\bibinfo
  {journal} {Phys. Rev. A}\ }\textbf {\bibinfo {volume} {96}},\ \bibinfo
  {pages} {062102} (\bibinfo {year} {2017})}\BibitemShut {NoStop}%
\bibitem [{\citenamefont {Nielsen}\ and\ \citenamefont
  {Chuang}(2011)}]{nielsen}%
  \BibitemOpen
  \bibfield  {author} {\bibinfo {author} {\bibfnamefont {M.~A.}\ \bibnamefont
  {Nielsen}}\ and\ \bibinfo {author} {\bibfnamefont {I.~L.}\ \bibnamefont
  {Chuang}},\ }\href@noop {} {\emph {\bibinfo {title} {Quantum Computation and
  Quantum Information: 10th Anniversary Edition}}},\ \bibinfo {edition} {10th}\
  ed.\ (\bibinfo  {publisher} {Cambridge University Press},\ \bibinfo {address}
  {USA},\ \bibinfo {year} {2011})\BibitemShut {NoStop}%
\bibitem [{\citenamefont {Alicki}\ and\ \citenamefont {Fannes}(2004)}]{alicki}%
  \BibitemOpen
  \bibfield  {author} {\bibinfo {author} {\bibfnamefont {R.}~\bibnamefont
  {Alicki}}\ and\ \bibinfo {author} {\bibfnamefont {M.}~\bibnamefont
  {Fannes}},\ }\href {\doibase 10.1088/0305-4470/37/5/l01} {\bibfield
  {journal} {\bibinfo  {journal} {Journal of Physics A: Mathematical and
  General}\ }\textbf {\bibinfo {volume} {37}},\ \bibinfo {pages} {L55}
  (\bibinfo {year} {2004})}\BibitemShut {NoStop}%
\bibitem [{\citenamefont {Winter}(2016)}]{winter}%
  \BibitemOpen
  \bibfield  {author} {\bibinfo {author} {\bibfnamefont {A.}~\bibnamefont
  {Winter}},\ }\href {\doibase 10.1007/s00220-016-2609-8} {\bibfield  {journal}
  {\bibinfo  {journal} {Communications in Mathematical Physics}\ }\textbf
  {\bibinfo {volume} {347}},\ \bibinfo {pages} {291} (\bibinfo {year}
  {2016})}\BibitemShut {NoStop}%
\bibitem [{\citenamefont {Holmes}(1975)}]{hahn}%
  \BibitemOpen
  \bibfield  {author} {\bibinfo {author} {\bibfnamefont {R.~B.}\ \bibnamefont
  {Holmes}},\ }\href {\doibase 10.1007/978-1-4684-9369-6} {\emph {\bibinfo
  {title} {Geometric Functional Analysis and its Applications}}}\ (\bibinfo
  {publisher} {Springer New York},\ \bibinfo {year} {1975})\BibitemShut
  {NoStop}%
\bibitem [{\citenamefont {Wilde}(2013)}]{wilde}%
  \BibitemOpen
  \bibfield  {author} {\bibinfo {author} {\bibfnamefont {M.~M.}\ \bibnamefont
  {Wilde}},\ }\href {\doibase 10.1017/CBO9781139525343} {\emph {\bibinfo
  {title} {Quantum Information Theory}}}\ (\bibinfo  {publisher} {Cambridge
  University Press},\ \bibinfo {year} {2013})\BibitemShut {NoStop}%
\bibitem [{\citenamefont {Cerf}\ and\ \citenamefont {Adami}(1999)}]{condamp}%
  \BibitemOpen
  \bibfield  {author} {\bibinfo {author} {\bibfnamefont {N.~J.}\ \bibnamefont
  {Cerf}}\ and\ \bibinfo {author} {\bibfnamefont {C.}~\bibnamefont {Adami}},\
  }\href {\doibase 10.1103/PhysRevA.60.893} {\bibfield  {journal} {\bibinfo
  {journal} {Phys. Rev. A}\ }\textbf {\bibinfo {volume} {60}},\ \bibinfo
  {pages} {893} (\bibinfo {year} {1999})}\BibitemShut {NoStop}%
\bibitem [{\citenamefont {Friis}\ \emph {et~al.}(2017)\citenamefont {Friis},
  \citenamefont {Bulusu},\ and\ \citenamefont {Bertlmann}}]{cond_geo}%
  \BibitemOpen
  \bibfield  {author} {\bibinfo {author} {\bibfnamefont {N.}~\bibnamefont
  {Friis}}, \bibinfo {author} {\bibfnamefont {S.}~\bibnamefont {Bulusu}}, \
  and\ \bibinfo {author} {\bibfnamefont {R.~A.}\ \bibnamefont {Bertlmann}},\
  }\href {\doibase 10.1088/1751-8121/aa5dfd} {\bibfield  {journal} {\bibinfo
  {journal} {Journal of Physics A: Mathematical and Theoretical}\ }\textbf
  {\bibinfo {volume} {50}},\ \bibinfo {pages} {125301} (\bibinfo {year}
  {2017})}\BibitemShut {NoStop}%
\bibitem [{Note1()}]{Note1}%
  \BibitemOpen
  \bibinfo {note} {This is because $\protect \qopname \relax o{log}0$ is
  undefined}\BibitemShut {NoStop}%
\bibitem [{\citenamefont {Werner}(1989)}]{werner}%
  \BibitemOpen
  \bibfield  {author} {\bibinfo {author} {\bibfnamefont {R.~F.}\ \bibnamefont
  {Werner}},\ }\href {\doibase 10.1103/physreva.40.4277} {\bibfield  {journal}
  {\bibinfo  {journal} {Physical Review A}\ }\textbf {\bibinfo {volume} {40}},\
  \bibinfo {pages} {4277} (\bibinfo {year} {1989})}\BibitemShut {NoStop}%
\bibitem [{\citenamefont {G\"{u}hne}\ and\ \citenamefont
  {T{\'{o}}th}(2009)}]{guhne}%
  \BibitemOpen
  \bibfield  {author} {\bibinfo {author} {\bibfnamefont {O.}~\bibnamefont
  {G\"{u}hne}}\ and\ \bibinfo {author} {\bibfnamefont {G.}~\bibnamefont
  {T{\'{o}}th}},\ }\href {\doibase 10.1016/j.physrep.2009.02.004} {\bibfield
  {journal} {\bibinfo  {journal} {Physics Reports}\ }\textbf {\bibinfo {volume}
  {474}},\ \bibinfo {pages} {1} (\bibinfo {year} {2009})}\BibitemShut {NoStop}%
\bibitem [{\citenamefont {Grant}\ and\ \citenamefont {Boyd}(2014)}]{cvx_1}%
  \BibitemOpen
  \bibfield  {author} {\bibinfo {author} {\bibfnamefont {M.}~\bibnamefont
  {Grant}}\ and\ \bibinfo {author} {\bibfnamefont {S.}~\bibnamefont {Boyd}},\
  }\href@noop {} {\enquote {\bibinfo {title} {{CVX}: Matlab software for
  disciplined convex programming, version 2.1},}\ }\bibinfo {howpublished}
  {\url{http://cvxr.com/cvx}} (\bibinfo {year} {2014})\BibitemShut {NoStop}%
\bibitem [{\citenamefont {Grant}\ and\ \citenamefont {Boyd}(2008)}]{cvx_2}%
  \BibitemOpen
  \bibfield  {author} {\bibinfo {author} {\bibfnamefont {M.}~\bibnamefont
  {Grant}}\ and\ \bibinfo {author} {\bibfnamefont {S.}~\bibnamefont {Boyd}},\
  }in\ \href@noop {} {\emph {\bibinfo {booktitle} {Recent Advances in Learning
  and Control}}},\ \bibinfo {series and number} {Lecture Notes in Control and
  Information Sciences},\ \bibinfo {editor} {edited by\ \bibinfo {editor}
  {\bibfnamefont {V.}~\bibnamefont {Blondel}}, \bibinfo {editor} {\bibfnamefont
  {S.}~\bibnamefont {Boyd}}, \ and\ \bibinfo {editor} {\bibfnamefont
  {H.}~\bibnamefont {Kimura}}}\ (\bibinfo  {publisher} {Springer-Verlag
  Limited},\ \bibinfo {year} {2008})\ pp.\ \bibinfo {pages} {95--110},\
  \bibinfo {note} {\url{http://stanford.edu/~boyd/graph_dcp.html}}\BibitemShut
  {NoStop}%
\bibitem [{\citenamefont {Cubitt}(2013)}]{quantinf}%
  \BibitemOpen
  \bibfield  {author} {\bibinfo {author} {\bibfnamefont {T.}~\bibnamefont
  {Cubitt}},\ }\href@noop {} {\enquote {\bibinfo {title} {Quantinf package,
  version 0.5.1},}\ }\bibinfo {howpublished}
  {\url{https://www.dr-qubit.org/matlab.html}} (\bibinfo {year}
  {2013})\BibitemShut {NoStop}%
\bibitem [{\citenamefont {Fawzi}\ \emph {et~al.}(2018)\citenamefont {Fawzi},
  \citenamefont {Saunderson},\ and\ \citenamefont {Parrilo}}]{cvxquad}%
  \BibitemOpen
  \bibfield  {author} {\bibinfo {author} {\bibfnamefont {H.}~\bibnamefont
  {Fawzi}}, \bibinfo {author} {\bibfnamefont {J.}~\bibnamefont {Saunderson}}, \
  and\ \bibinfo {author} {\bibfnamefont {P.~A.}\ \bibnamefont {Parrilo}},\
  }\href@noop {} {\bibfield  {journal} {\bibinfo  {journal} {Foundations of
  Computational Mathematics}\ } (\bibinfo {year} {2018})},\ \bibinfo {note}
  {package cvxquad at \url{https://github.com/hfawzi/cvxquad}}\BibitemShut
  {NoStop}%
\bibitem [{\citenamefont {Wolf}(2003)}]{partial_transpose}%
  \BibitemOpen
  \bibfield  {author} {\bibinfo {author} {\bibfnamefont {M.~M.}\ \bibnamefont
  {Wolf}},\ }\href@noop {} {\emph {\bibinfo {title} {Partial transposition in
  quantum information theory}}}\ (\bibinfo {year} {2003})\BibitemShut {NoStop}%
\bibitem [{\citenamefont {Bertlmann}\ and\ \citenamefont
  {Krammer}(2009)}]{bertlmann}%
  \BibitemOpen
  \bibfield  {author} {\bibinfo {author} {\bibfnamefont {R.~A.}\ \bibnamefont
  {Bertlmann}}\ and\ \bibinfo {author} {\bibfnamefont {P.}~\bibnamefont
  {Krammer}},\ }\href {\doibase 10.1016/j.aop.2009.01.008} {\bibfield
  {journal} {\bibinfo  {journal} {Annals of Physics}\ }\textbf {\bibinfo
  {volume} {324}},\ \bibinfo {pages} {1388} (\bibinfo {year}
  {2009})}\BibitemShut {NoStop}%
\bibitem [{\citenamefont {Bertlmann}\ \emph {et~al.}(2002)\citenamefont
  {Bertlmann}, \citenamefont {Narnhofer},\ and\ \citenamefont
  {Thirring}}]{ref_14}%
  \BibitemOpen
  \bibfield  {author} {\bibinfo {author} {\bibfnamefont {R.~A.}\ \bibnamefont
  {Bertlmann}}, \bibinfo {author} {\bibfnamefont {H.}~\bibnamefont
  {Narnhofer}}, \ and\ \bibinfo {author} {\bibfnamefont {W.}~\bibnamefont
  {Thirring}},\ }\href {\doibase 10.1103/PhysRevA.66.032319} {\bibfield
  {journal} {\bibinfo  {journal} {Phys. Rev. A}\ }\textbf {\bibinfo {volume}
  {66}},\ \bibinfo {pages} {032319} (\bibinfo {year} {2002})}\BibitemShut
  {NoStop}%
\bibitem [{\citenamefont {Bertlmann}\ \emph {et~al.}(2005)\citenamefont
  {Bertlmann}, \citenamefont {Durstberger}, \citenamefont {Hiesmayr},\ and\
  \citenamefont {Krammer}}]{ref_25}%
  \BibitemOpen
  \bibfield  {author} {\bibinfo {author} {\bibfnamefont {R.~A.}\ \bibnamefont
  {Bertlmann}}, \bibinfo {author} {\bibfnamefont {K.}~\bibnamefont
  {Durstberger}}, \bibinfo {author} {\bibfnamefont {B.~C.}\ \bibnamefont
  {Hiesmayr}}, \ and\ \bibinfo {author} {\bibfnamefont {P.}~\bibnamefont
  {Krammer}},\ }\href {\doibase 10.1103/PhysRevA.72.052331} {\bibfield
  {journal} {\bibinfo  {journal} {Phys. Rev. A}\ }\textbf {\bibinfo {volume}
  {72}},\ \bibinfo {pages} {052331} (\bibinfo {year} {2005})}\BibitemShut
  {NoStop}%
\bibitem [{\citenamefont {Wang}\ \emph {et~al.}(2015)\citenamefont {Wang},
  \citenamefont {Xu}, \citenamefont {Campbell},\ and\ \citenamefont
  {Severini}}]{weakly_optimal}%
  \BibitemOpen
  \bibfield  {author} {\bibinfo {author} {\bibfnamefont {B.-H.}\ \bibnamefont
  {Wang}}, \bibinfo {author} {\bibfnamefont {H.-R.}\ \bibnamefont {Xu}},
  \bibinfo {author} {\bibfnamefont {S.}~\bibnamefont {Campbell}}, \ and\
  \bibinfo {author} {\bibfnamefont {S.}~\bibnamefont {Severini}},\ }\href@noop
  {} {\bibfield  {journal} {\bibinfo  {journal} {Quantum Info. Comput.}\
  }\textbf {\bibinfo {volume} {15}},\ \bibinfo {pages} {1109–1121} (\bibinfo
  {year} {2015})}\BibitemShut {NoStop}%
\bibitem [{\citenamefont {Berta}\ \emph {et~al.}(2012)\citenamefont {Berta},
  \citenamefont {Fawzi},\ and\ \citenamefont {Wehner}}]{randomness_extractors}%
  \BibitemOpen
  \bibfield  {author} {\bibinfo {author} {\bibfnamefont {M.}~\bibnamefont
  {Berta}}, \bibinfo {author} {\bibfnamefont {O.}~\bibnamefont {Fawzi}}, \ and\
  \bibinfo {author} {\bibfnamefont {S.}~\bibnamefont {Wehner}},\ }in\ \href
  {\doibase 10.1007/978-3-642-32009-5_45} {\emph {\bibinfo {booktitle} {Lecture
  Notes in Computer Science}}}\ (\bibinfo  {publisher} {Springer Berlin
  Heidelberg},\ \bibinfo {year} {2012})\ pp.\ \bibinfo {pages}
  {776--793}\BibitemShut {NoStop}%
\bibitem [{\citenamefont {Bennett}\ \emph {et~al.}(1996)\citenamefont
  {Bennett}, \citenamefont {DiVincenzo}, \citenamefont {Smolin},\ and\
  \citenamefont {Wootters}}]{distillation}%
  \BibitemOpen
  \bibfield  {author} {\bibinfo {author} {\bibfnamefont {C.~H.}\ \bibnamefont
  {Bennett}}, \bibinfo {author} {\bibfnamefont {D.~P.}\ \bibnamefont
  {DiVincenzo}}, \bibinfo {author} {\bibfnamefont {J.~A.}\ \bibnamefont
  {Smolin}}, \ and\ \bibinfo {author} {\bibfnamefont {W.~K.}\ \bibnamefont
  {Wootters}},\ }\href {\doibase 10.1103/PhysRevA.54.3824} {\bibfield
  {journal} {\bibinfo  {journal} {Phys. Rev. A}\ }\textbf {\bibinfo {volume}
  {54}},\ \bibinfo {pages} {3824} (\bibinfo {year} {1996})}\BibitemShut
  {NoStop}%
\bibitem [{\citenamefont {Gyongyosi}\ \emph {et~al.}(2018)\citenamefont
  {Gyongyosi}, \citenamefont {Imre},\ and\ \citenamefont {Nguyen}}]{coherent}%
  \BibitemOpen
  \bibfield  {author} {\bibinfo {author} {\bibfnamefont {L.}~\bibnamefont
  {Gyongyosi}}, \bibinfo {author} {\bibfnamefont {S.}~\bibnamefont {Imre}}, \
  and\ \bibinfo {author} {\bibfnamefont {H.~V.}\ \bibnamefont {Nguyen}},\
  }\href {\doibase 10.1109/COMST.2017.2786748} {\bibfield  {journal} {\bibinfo
  {journal} {IEEE Communications Surveys Tutorials}\ }\textbf {\bibinfo
  {volume} {20}},\ \bibinfo {pages} {1149} (\bibinfo {year}
  {2018})}\BibitemShut {NoStop}%
\bibitem [{\citenamefont {{Smith}}(2010)}]{capacity}%
  \BibitemOpen
  \bibfield  {author} {\bibinfo {author} {\bibfnamefont {G.}~\bibnamefont
  {{Smith}}},\ }\href@noop {} {\bibfield  {journal} {\bibinfo  {journal} {arXiv
  e-prints}\ ,\ \bibinfo {eid} {arXiv:1007.2855}} (\bibinfo {year} {2010})},\
  \Eprint {http://arxiv.org/abs/1007.2855} {arXiv:1007.2855 [cs.IT]}
  \BibitemShut {NoStop}%
\bibitem [{\citenamefont {Robertson}(1929)}]{ur}%
  \BibitemOpen
  \bibfield  {author} {\bibinfo {author} {\bibfnamefont {H.~P.}\ \bibnamefont
  {Robertson}},\ }\href {\doibase 10.1103/physrev.34.163} {\bibfield  {journal}
  {\bibinfo  {journal} {Physical Review}\ }\textbf {\bibinfo {volume} {34}},\
  \bibinfo {pages} {163} (\bibinfo {year} {1929})}\BibitemShut {NoStop}%
\bibitem [{\citenamefont {Kraus}(1987)}]{complementary}%
  \BibitemOpen
  \bibfield  {author} {\bibinfo {author} {\bibfnamefont {K.}~\bibnamefont
  {Kraus}},\ }\href {\doibase 10.1103/PhysRevD.35.3070} {\bibfield  {journal}
  {\bibinfo  {journal} {Phys. Rev. D}\ }\textbf {\bibinfo {volume} {35}},\
  \bibinfo {pages} {3070} (\bibinfo {year} {1987})}\BibitemShut {NoStop}%
\bibitem [{\citenamefont {Maassen}\ and\ \citenamefont
  {Uffink}(1988)}]{generalized_uncertainty}%
  \BibitemOpen
  \bibfield  {author} {\bibinfo {author} {\bibfnamefont {H.}~\bibnamefont
  {Maassen}}\ and\ \bibinfo {author} {\bibfnamefont {J.~B.~M.}\ \bibnamefont
  {Uffink}},\ }\href {\doibase 10.1103/PhysRevLett.60.1103} {\bibfield
  {journal} {\bibinfo  {journal} {Phys. Rev. Lett.}\ }\textbf {\bibinfo
  {volume} {60}},\ \bibinfo {pages} {1103} (\bibinfo {year}
  {1988})}\BibitemShut {NoStop}%
\end{thebibliography}

\end{document}